\documentclass{article}
\usepackage[utf8]{inputenc}
\usepackage{geometry}
\usepackage{array}
\usepackage{multirow}
\usepackage{amsmath}
\usepackage{amsthm}
\usepackage{amssymb}
\usepackage[numbers,sort]{natbib}

\usepackage{color}

\usepackage{authblk}

\usepackage{newtxtext}

\usepackage[pdfencoding=auto, psdextra]{hyperref}
\makeatletter

\providecommand{\tabularnewline}{\\}

\theoremstyle{plain}
\newtheorem{thm}{\protect\theoremname}[section]
\theoremstyle{remark}
\newtheorem{rem}[thm]{\protect\remarkname}
\theoremstyle{plain}
\newtheorem{prop}[thm]{\protect\propositionname}
\theoremstyle{plain}
\newtheorem{lem}[thm]{\protect\lemmaname}
\theoremstyle{plain}
\newtheorem{ass}[thm]{\protect\assumptionname}

\theoremstyle{definition}
\newtheorem{defn}[thm]{Definition}

\makeatother

\providecommand{\lemmaname}{Lemma}
\providecommand{\propositionname}{Proposition}
\providecommand{\remarkname}{Remark}
\providecommand{\theoremname}{Theorem}
\providecommand{\assumptionname}{Assumption}

\begin{document}
	
	\title{On the time dependence of\\the rate of convergence towards Hartree
		dynamics \\for interacting Bosons}
	
	\author{Jinyeop Lee\thanks{\url{jinyeoplee@kaist.ac.kr}}}
	\affil{Department of Mathematical Sciences, KAIST}
	\date{}
	\maketitle
	\begin{abstract}
		We consider interacting $N$-Bosons in three dimensions. It is known
		that the difference between the many-body Schr\"odinger evolution
		in the mean-field regime and the corresponding Hartree dynamics is
		of order $1/N$. We investigate the time dependence of the difference.
		To have sub-exponential bound, we use the results of time decay estimate for small initial data. We also refine time dependent bound for singular potential using Strichartz estimate.
		We consider the interaction potential $V(x)$ of type $\lambda\exp(-\mu|x|)|x|^{-\gamma}$
		for $\lambda\in\mathbb{R}$, $\mu\geq0$, and $0<\gamma<3/2$, which
		covers the Coulomb and Yukawa interaction.
	\end{abstract}
	
	\section{Introduction and the main results} \label{sec:intro}
	
	We consider a many-body particle system of $N$-Bosons with two body
	interaction via Coulomb type interaction or Yukawa type interaction,
	i.e, $V(x)=\lambda\exp(-\mu|x|)|x|^{-\gamma}$ with $\lambda\in\mathbb{R}$,
	$\mu\geq0$, and $0<\gamma<3/2$. The system can be described by a complex
	valued function $\psi_{N}=\psi_{N}(x_{1},\dots,x_{N}):(\mathbb{R}^{3})^{N}\to\mathbb{C}$,
	which is called wave function. The wave function $\psi_{N}$ for the Bosonic system is symmetric under the permutation of variables, i.e., for each $x_{i},x_{j}\in\mathbb{R}^{3}$ $1\leq i,j\leq N$,
	$\psi_{N}(\dots,x_{j},\dots,x_{i},\dots)=\psi_{N}(\dots,x_{i},\dots,x_{j},\dots)$.
	Our system is governed by the following Hamiltonian:
	\begin{equation}
	H_{N}=\sum_{j=1}^{N}-\Delta_{j}+\frac{1}{N-1}\sum_{i<j}V(x_{i}-x_{j}),\label{eq:N_body_Hamiltonian}
	\end{equation}
	and we call it a many-body mean-field Hamiltonian. 
	
	Now, suppose that the system is fully condensed, i.e., the initial
	wave function is given by
	\[
	\psi_{N}=\varphi^{\otimes N}
	\]
	with a one-body wave function $\varphi:\mathbb{R}^{3}\to\mathbb{C}$
	in some appropriate function space which will be described later.
	We want to argue that the system is almost condensed at the time $t\geq0$
	as well, i.e, 
	\begin{equation}
	\psi_{N,t}=e^{-\mathrm{i}H_{N}t}\psi_{N}\simeq\varphi_{t}^{\otimes N}\quad\text{for large }N\label{eq:Factorization_of_Wave}
	\end{equation}
	for some $\varphi_{t}:\mathbb{R}^{3}\to\mathbb{C}$. 
	
	Heuristically, from the point of view of particle $x_1$, it `feels' averaged potential 
	\[
	\frac{1}{N-1}\sum_{j=2}^{N}V(x_{1}-x_{j})
	\]
	from other particles.
	Since the Hamiltonian is symmetric under the permutation of the particles, the averaged potential is the same for every particle $x_j$.
	Thus, we can expect that $\varphi_{t}$ evolves according to the Hartree equation
	\begin{equation}
	\mathrm{i}\partial_{t}\varphi_{t}=-\Delta\varphi_{t}+(V*|\varphi_{t}|^{2})\varphi_{t}\label{eq:Hartree}
	\end{equation}
	with initial data $\varphi_{t=0}=\varphi$. Non-rigorous derivation of the Hartree equation can be found in literature. (See, e.g., Section 1 of~\cite{Chen2018}).
	
	To understand the `almost condensation' of the system at the time $t\geq0$ in a mathematically rigorous way, we proceed as follows.
	First, we consider the density matrix $\gamma_{N,t}=\left|\psi_{N,t}\right\rangle \left\langle \psi_{N,t}\right|$
	associated with $\psi_{N,t}$, which can be understood as the orthogonal projection
	onto $\psi_{N,t}$. More precisely, the kernel of $\gamma_{N,t}$ is given by
	\[
	\gamma_{N,t}(\mathbf{x};\mathbf{x}')=\psi_{N,t}(\mathbf{x})\overline{\psi_{N,t}(\mathbf{x})}.
	\]
	The $k$-particle marginal density is then defined through its kernel
	\begin{equation}
	\gamma_{N,t}^{\left(k\right)}(\mathbf{x}_{k};\mathbf{x}'_{k})=\int\mathrm{d}\mathbf{x}_{N-k}\gamma_{N,t}(\mathbf{x}_{k},\mathbf{x}_{N-k};\mathbf{x}'_{k},\mathbf{x}_{N-k}). \label{eq:Kernel_of_Marginal_Density}
	\end{equation}
	We now focus on the trace-norm distance between the one-particle marginal density $\gamma_{N,t}^{(1)}$ and the projection operator $|\varphi_{t}\rangle\langle\varphi_{t}|$.
	In particular,
	we will prove that
	\begin{equation}
	\operatorname{Tr}\left|\gamma_{N,t}^{(1)}-|\varphi_{t}\rangle\langle\varphi_{t}|\right|\leq\frac{C(t)}{N}\label{eq:trace_norm_bound}
	\end{equation}
	and find $C(t)$ according to the conditions on $V$. It is known that
	the optimal $N$-dependence for the rate of convergence is of
	$O(1/N)$. (See, e.g., \cite{Chen2011,Chen2018,Grillakis13,Grillakis17}.) 
	For the necessity of the trace-norm in \eqref{eq:trace_norm_bound}, we refer to \cite{knowles2010mean}, where it is also provided an example that explains why $L^2$-norm is counterintuitive.
	Moreover, if the initial many-body state is fully factorized,
	for every $t>0$, the evolved state is never close to the state $\varphi^{\otimes N}$ in the $L^2$-norm, except
	in the non interacting case. One can quote in this contest the several works aimed to find a norm-approximation of the many-body evolution, by taking into account fluctuations
	around the Hartree dynamics, see for example \cite{chen2012,Grillakis2010,Grillakis2011,Grillakis13,lewin2013,Rodnianski2009}, and the pioneering papers
	by Hepp and Ginibre-Velo \cite{Ginibre1979_1,Ginibre1979_2,Hepp1974}.
	
	Historically, Spohn \cite{Spohn1980} first proved that $\operatorname{Tr}\left|\gamma_{N,t}^{(1)}-|\varphi_{t}\rangle\langle\varphi_{t}|\right|\to0$
	as $N\to\infty$ for bounded potential. It was extended by Erd\H{o}s and Yau \cite{Erdos2001} to prove the same result for singular potential (including the Coulomb case) by using the BBGKY hierarchy. The rate of convergence, especially the $N$-dependence of the bound in \eqref{eq:trace_norm_bound}, has been intensively studied in last ten years. 
	First, a new method based on coherent state
	approach was introduced by Rodnianski and Schlein in \cite{Rodnianski2009} to give an explicit rate of convergence as in \eqref{eq:trace_norm_bound} with an $O(1/\sqrt{N})$ bound.
	The proof is based on the Fock space approach that was introduced by Hepp \cite{Hepp1974} and extended by Ginibre and Velo \cite{Ginibre1979_1,Ginibre1979_2}.
	Soon after \cite{Rodnianski2009}, Knowles and Pickl \cite{knowles2010mean} considered more singular
	interaction potentials and obtained similar estimates on the rate of convergence.
	The proof in \cite{knowles2010mean} is based on the use of projection operators in the $N$-particle space $L_s^2(\mathbb{R}^{3N})$, and allows for a large class of possibly time-dependent external potentials.
	The $O(1/N)$ rate of convergence, which is optimal in $N$-dependence, was proved by 
	Chen, Lee, and Schlein in \cite{Chen2011a} for the Coulomb case. It was later extended in \cite{Chen2018} to cover the case $V\in L^{2}+L^{\infty}$.
	We also remark that Hott \cite{Hott2018} pointed out the initial
	condition may stay in bigger space than $H^{1}(\mathbb{R}^{3})$ for
	$V(x)=\lambda|x|^{-\gamma}$ with $1<\gamma<3/2$.
	
	Unlike the $N$-dependence in the rate of the convergence, the time dependence of the bound has mostly been of order $e^{Kt}$ (or even worse) in most of the works mentioned above, with the exception of \cite{knowles2010mean}
	where the authors also showed that the rate of convergence can be uniform in time if the solution of the Hartree equation
	satisfies an integrability condition. For example, in \cite{Chen2018}, where the use of Strichartz estimates\footnote{For the details of the Strichartz estimate, see \cite{cazenave2003semilinear,tao2006nonlinear}} was the main strategy of the proof to generalize the interaction potential, the time dependence of the bound is of order $e^{Kt^{3/2}}$ which grows faster than $e^{Kt}$.
	
	It is in general harder to obtain the better bound in terms of $N$-dependence in \eqref{eq:trace_norm_bound} for more singular interaction potential, e.g. $\gamma>1$. On the contrary, it is typically more difficult to prove the better bound in terms of $t$-dependence for slowly decaying interaction potential or long range potential, e.g. $\gamma<1$. (Heuristically, we can also argue that the optimal bound can be proved relatively easily, since it decays sufficiently fast.) It was also remarked by Knowles and Pickl in \cite{knowles2010mean} that the time dependence
	can be removed for interaction potential with strong decay.
	Such a phenomenon is known in the Gross-Pitaevskii regime; we refer to the work of Chong \cite{Chong2016}, where the scattering results of the cubic nonlinear Schr\"odinger equation were used.
	For a inverse power law potential $V(x)=\lambda|x|^{-\gamma}$ with $0<\gamma<3/2$, however, one may not have the corresponding scattering result for the solution. Nevertheless, it is possible to use the Strichartz estimates as in \cite{Chen2018}, which can be regarded as a generalized time decay estimate in the time averaged sense. We remark that the time decay estimates of the Hartree equation has been deeply researched in many important works by Hayashi, Naumkin, and Ozawa \cite{HayashiNaumkin98,Hayashi01,HayashiOzawa87}. Moreover, the existence of the modified operator of the equation was studied, e.g., by Nakanishi \cite{Nakanishi02-2,Nakanishi02-1}.
	
	A similar approach can also be applied to many-body semi-relativistic Schr\"odinger equations which describes a Boson star. Lee \cite{Lee2013} provide the optimal rate of convergence $O(1/N)$ for Coulomb interaction. Following the approach presented in this article, it is believed that one can obtain a corresponding bound for the semi-relativistic case by exploiting the properties of the mean-field solution. We refer to the work of Cho and Ozawa \cite{cho2006} for more detail on the solution of the semi-relativistic Hartree equation. The time dependence of the bound in the semi-relativistic case will be discussed in a future paper.
	
	In this article, we investigate the time dependence $C(t)$ in \eqref{eq:trace_norm_bound} by using the results of time decay estimates and Strichartz estimates for $V(x)=\lambda\exp(-\mu|x|)|x|^{-\gamma}$
	for $\lambda\in\mathbb{R}$, $\mu\geq0$ and $0<\gamma<3/2$. 
	More precisely, we prove that the bound in \eqref{eq:trace_norm_bound} is time-independent if the interaction constant is below a threshold, i.e., $|\lambda|<\lambda_c$ for some $\lambda_c =\lambda_c(\gamma,\mu)$. We also improve the time dependence on the bound for more singular potential with $1<\gamma<3/2$ and $\alpha\in[2\gamma/3,1)$ to $C_\alpha e^{Kt^{\gamma/\alpha}}$, which was $Ce^{Kt^{3/2}}$ in \cite{Chen2018}. For the exact Coulomb interaction case with $\gamma = 1$, we prove a bound that is a polynomial of $t$ whose degree is proportional to $\lambda$, hence sublinear in $t$ if $\lambda$ is sufficiently small. The bounds are collected in Table \ref{table:RoC_t}, which describes
	the time dependence of the rate of convergence. 
	
	\textbf{Notational Remark.}
	We use $\|f\|_{p}=\|f\|_{L^{p}(X)}$ for the standard
	$L^{p}$ norm of $f:X\to\mathbb{C}$.
	We also use $\|f\|_{H^p}=\|f\|_{H^{p}(X)}$ for the standard
	$H^{p}$ norm of $f:X\to\mathbb{C}$.
	We denote $\|J\|_\mathrm{op}$ as an operator norm of an operator $J$.	
	In many lines of
	inequalities we will face constants $C$ here and there, note that
	the constant may differ line by line. The time dependent constant
	$C(t)$ also can differ line by line. Sometimes we may use
	$C_{\alpha}$ if we want to emphasize the dependence on a variable
	$\alpha$.
	We write $\mathcal{S}$ to denote the Schwartz space and $\mathcal{S}'$ to denote the dual space of $\mathcal{S}$.
	
	\begin{defn}
		We define a generalized Sobolev space, or weighted Sobolev
		space, such that
		\[
		H_{p}^{m,s}=\left\{ \phi\in\mathcal{S}':\|\phi\|_{m,s,p}=\|(1+|x|^{2})^{s/2}(1-\Delta)^{m/2}\phi\|_{p}<\infty\right\} 
		\]
		for $m,s\in\mathbb{R}$. We may simply write $H^{m,s}$ to denote
		$H_{2}^{m,s}$.
	\end{defn}
	
	Note that $H^{s,0}=H^{s}$ and $\varphi\in H^{0,k}$ implies $\widetilde{\varphi}\in H^{k}$.
	Moreover, because one can think of $|\varphi|^{2}$ as a probability
	distribution under normalization, if $\varphi\in H^{0,\gamma}$, one
	can understand that the $\gamma$-th moment of $|\varphi|^{2}$ is
	finite. 
	
	\begin{ass}\label{ass:cases} We assume initial data $\varphi$ for given $\lambda$, $\gamma$, and $\mu$ such that
		\begin{enumerate}
			\item for $|\lambda|\leq\lambda_c$ and $0<\gamma<1$, let $\varphi\in H^{5,0}\cap H^{0,5}$
			with $\|\varphi\|_{H^{5,0}}+\|\varphi\|_{H^{0,5}}=1$, 
			\item for $|\lambda|\leq\lambda_c$ and $1\leq\gamma<3/2$, let $\varphi\in H^{S,0}\cap H^{0,S}$
			with $\|\varphi\|_{H^{S,0}}+\|\varphi\|_{H^{0,S}}=1$ for $S>3/2$, 
			\item for $|\lambda|\leq\lambda_c$ and $\mu>0$, let $\varphi\in H^{S,0}\cap H^{0,S}$
			with $\|\varphi\|_{H^{S,0}}+\|\varphi\|_{H^{0,S}}=1$ for $S>3/2$, 
			\item for $\lambda>\lambda_c$, $\mu=0$ and $1\leq\gamma<3/2$, let $\varphi\in H^{2,0}\cap H^{0,2}$, 
			or
			\item for $\lambda>\lambda_c$ and $\mu>0$, let $\varphi\in H^{2,0}\cap H^{0,2}$.
			\item otherwise, let $\varphi\in H^{1}(\mathbb{R}^{3})$,
		\end{enumerate}
	\end{ass}
	
	\begin{thm}
		Assume that the potential $V(x)=\lambda\exp(-\mu|x|)|x|^{-\gamma}$ 
		with
		interaction constant $\lambda\in\mathbb{R}$ and
		positive $\mu\geq0$. Let $\lambda_c=\lambda_c(\mu,\gamma)$ be a threshold of interaction constant.
		\label{thm:classical_Main}
		Assume that $\varphi$ follows the Assumption \ref{ass:cases} for each case.
		Let $\varphi_{t}$ be the solution of the Hartree equation
		\[
		\mathrm{i}\partial_{t}\varphi_{t}=-\Delta\varphi_{t}+(V*|\varphi_{t}|^{2})\varphi_{t}
		\]
		with initial data $\varphi_{t=0}=\varphi$. Let $\psi_{N,t}=e^{-\mathrm{i}H_{N}t}\varphi^{\otimes N}$
		and $\gamma_{N,t}^{(1)}$ be the one-particle reduced density associated
		with $\psi_{N,t}$, as defined in \eqref{eq:Kernel_of_Marginal_Density}.
		Then there exists a time-dependent constant $C(t)$, depending only on $\varphi$,
		$\lambda$, $\mu$, and $t$ such that
		\begin{equation}
		\operatorname{Tr}\left|\gamma_{N,t}^{(1)}-|\varphi_{t}\rangle\langle\varphi_{t}|\right|\leq\frac{C(t)}{N}.\label{eq:rate_of_Conv_CN}
		\end{equation}
		Moreover, we can choose the time dependent factor $C(t)$ in \eqref{eq:rate_of_Conv_CN} as in the Table \ref{table:RoC_t} with constants $C$ and $K$ independent of $t$, arbitrary constant $\alpha\in[2\gamma/3,1)$, and $\lambda_c=\lambda_c(\mu,\gamma)$.
		\begin{center}
			\begin{table}[h]
				\caption{Time dependent factor $C(t)$ of the rate of convergence }
				\label{table:RoC_t}
				\centering{}
				\begin{tabular}{|c|c|c|c|c|}
					\hline 
					\multirow{2}{*}{} & \multicolumn{3}{c|}{$V(x)=\lambda|x|^{-\gamma}$} & \multicolumn{1}{c|}{$V(x)=\lambda\exp(-\mu|x|)|x|^{-\gamma}$, $\mu>0$}\tabularnewline
					\cline{2-5} 
					& $0<\gamma<1$ & $\gamma=1$ & $1<\gamma<3/2$ & $0<\gamma<3/2$\tabularnewline
					\hline 
					$\lambda>\lambda_c$ & $Ce^{Kt}$ & $Ce^{Kt^{1/3}}$ & $Ce^{Kt^{1-2\gamma/3}}$ & $C(1+t)^K$ \tabularnewline
					\hline 
					$|\lambda|\leq\lambda_c$ & $Ce^{Kt^{1-\gamma}}$ & $C(1+t)^{K}$ & $C$ & $C$ \tabularnewline
					\hline 
					$\lambda<-\lambda_c$ & $Ce^{Kt}$ & $Ce^{Kt}$ & $C_{\alpha}e^{Kt^{\gamma/\alpha}}$ & $Ce^{Kt}$\tabularnewline
					\hline 
				\end{tabular}
			\end{table}
			\par\end{center}
		
	\end{thm}
	
	\begin{rem}
		Note that for $V(x)=\lambda|x|^{-\gamma}$ with $\lambda<-\lambda_c$,
		according to \cite{Chen2018}, the exponent of $t$ was $3/2$ which
		is the case $\alpha=2\gamma/3$. The current paper provides a better
		time growth rate. 
	\end{rem}
	
	\begin{rem}
		Notice that in the case of Coulomb interaction the exponent $K$ in the bound $C(1+t)^K$ is sufficiently small, for small enough $\lambda$.
		In the proof, we show that $K$ is proportional to $|\lambda|$, i.e., $K=k|\lambda|$
		for fixed $k>0$. Because we are dealing with $|\lambda|<\lambda_c$ with small $\lambda_c$, $K=\kappa|\lambda|$
		is also sufficiently small for some constant $\kappa$. Thus, even though it is written as a polynomial of $(1+t)$, it is actually sublinear in $(1+t)$.
	\end{rem}
	
	\begin{rem}
		In \cite{knowles2010mean}, the authors remarked that if $\|\varphi_t\|_{q_1}$ and $\|\varphi_t\|_{q_2}$ is integrable in $t$ over $\mathbb{R}$, then the time dependent factor is uniform in time, i.e. $C(t)<\infty$, where $V\in L^{p_1}(\mathbb{R}^3)+L^{p_2}(\mathbb{R}^3)$ and $1/2 = 1/p_i + 1/q_i$ for $i=1,2$. 
		They also noted that such an integrability condition describes a scattering regime and it requires an interaction potential with strong decay.
		The result of the current article suggests that the strong decay of $V$, i.e., large $\gamma$, may not be enough to guarantee the scattering behavior but one actually needs to consider 
		the size of 
		the interaction constant $\lambda$.
		Intuitively, if the interaction constant is too large,
		the interactions between particles are hard to ignore even with strong decay.
		Thus,
		the particles cannot be asymptotically free even for large $t$, and one cannot expect the usual scattering behavior.
	\end{rem}
	
	\begin{rem}
		The result of Theorem \ref{thm:classical_Main} is expected to hold under more general assumptions on the many-body initial state, namely for initial states which exhibits condensation into a one particle orbital $\varphi$ (in the sense of the convergence of the one particle reduced density) but are not necessarily factorized. Unfortunately the method used in this paper is only relevant for factorized initial state. It would be interesting to see whether an improvement of the time dependence as in Theorem \ref{thm:classical_Main} could be also achieved using different methods, allowing for more general initial data. 
	\end{rem}
	
	We follow the approach in \cite{Chen2011,Chen2011a,Chen2018,Rodnianski2009} for the proof of Theorem \ref{thm:classical_Main}. In this method based on the analysis of the coherent states in the Fock space, the main obstacle is that a bound on the term 
	$\int_{0}^{t}\mathrm{d}s\,\|V(\cdot-x)\varphi_{s}\|_{2}$ is required. For this reason, we begin by establishing the time dependence of $\int_{0}^{t}\mathrm{d}s\,\|V(\cdot-x)\varphi_{s}\|_{2}$. The estimate is based on several time decay estimates of the solution of the Hartree equation.
	
	The rest of the paper is organized as follows: We will provide the estimates for $\int_{0}^{t}\mathrm{d}s\,\|V(\cdot-x)\varphi_{s}\|_{2}$ in Section 2. In Section 2.2, we will
	provide a sketch of proof of time decay estimates for Yukawa interaction,
	because it is a simple adjustment of previous results \cite{HayashiNaumkin98,HayashiOzawa87}.
	In Section 3, we briefly provide definitions and properties
	of Fock space which we are going to use. Section 4 is devoted to give
	proof of the main theorem. We have many useful bounds for operators in
	Fock space to prove the main theorem in Section 5. 
	While the most of the materials in Sections 3 through 5 are similar to those in the previous works \cite{Chen2011a,Chen2018,Rodnianski2009}, we do not omit them in the current paper in order to provide a logically complete explanation of our proof.
	
	\section{Properties of solution of mean-field equation}
	
	This section is devoted to provide time dependent or time independent
	bounds of $\int_{0}^{t}\mathrm{d}s\,\|V(\cdot-x)\varphi_{s}\|_{2}$
	for each case appeared in Table \ref{table:RoC_t}.
	
	\subsection{Time decay estimate of the Hartree equation for Coulomb type interaction\label{sec:prelim}}
	
	This section introduces time decay estimates of the Hartree equation.
	We will
	show that
	\[
	\int_{0}^{\infty}\mathrm{d}t\,\|V(\cdot-x)\varphi_{t}\|_{2}<C
	\]
	using time decay estimates for weakly attracting Hartree equation. 
	
	\begin{prop}
		\label{prop:timedecay}Suppose that $\varphi_{t}$ is a solution of
		\eqref{eq:Hartree}. Suppose that $\lambda$ in \eqref{eq:Hartree}
		is sufficiently small, the suitable size of $\lambda$ is depending
		on $\gamma$, $\mu$, and $\varphi$. We assume that
		\begin{enumerate}
			\item $\varphi\in H^{5,0}\cap H^{0,5}$ for $\mu=0$ and $0<\gamma<1$, 
			\item $\varphi\in H^{S,0}\cap H^{0,S}$ for $\mu=0$ and $1\leq\gamma<3/2$
			with $S>3/2$, or
			\item $\varphi\in H^{S,0}\cap H^{0,S}$ for $\mu>0$ and $0\leq\gamma<3/2$
			with $S>3/2$.
		\end{enumerate}
		Then there exists a unique global solution $\varphi_{t}$ of \eqref{eq:Hartree}
		such that 
		\[
		\|\varphi_{t}\|_{\infty}\leq C_{\lambda}(1+|t|)^{-3/2}.
		\]
	\end{prop}
	
	\begin{prop}
		\label{prop:timedecay_defocus}Suppose that $\varphi_{t}$ is a solution
		of \eqref{eq:Hartree}. We assume that $\varphi\in H^{2,0}\cap H^{0,2}$
		for (i) $\lambda>0$, $\mu=0$ and $1<\gamma<3/2$ or (ii) $\lambda>0$, $\mu>0$ and $0<\gamma<3/2$. Then there exists a unique global
		solution $\varphi_{t}$ of \eqref{eq:Hartree} such that 
		\[
		\|\varphi_{t}\|_{\infty}\leq C_{\lambda}(1+|t|)^{-1/2}.
		\]
	\end{prop}
	
	To prove this theorem, we are going to use small data scattering theory
	for Hartree dynamics; Hayashi and Namukin found that:
	\begin{lem}[Hayashi and Namukin 98']
		\label{prop:HayashiNaumkin98}We assume that $\varphi\in H^{S,0}(\mathbb{R}^{n})\cap H^{0,S}(\mathbb{R}^{n})$
		and $\|\varphi\|_{S,0}+\|\varphi\|_{0,S}=\epsilon'<\epsilon$, where
		$\epsilon$ is sufficiently small and $n/2<S<p=1+2/n$. Then there
		exists a unique global solution $\varphi_{t}$ to the Hartree equation
		\eqref{eq:Hartree}, with 
		\[
		V(x)=\lambda|x|^{-1}+\mu|x|^{-\delta}
		\]
		for $1<\delta<n$, such that 
		\[
		\varphi_{t}\in C(\mathbb{R},H^{S,0}\cap H^{0,S})
		\]
		and
		\[
		\|\varphi_{t}\|_{\infty}\leq C\epsilon'(1+|t|)^{-3/2}.
		\]
	\end{lem}
	
	\begin{proof}
		See Theorem 1.1 of \cite{HayashiNaumkin98}.
	\end{proof}
	
	\begin{lem}[Hayashi and Naumkin 01']
		\label{prop:HayashiNaumkin01}We assume that $\varphi\in H^{5,0}(\mathbb{R}^{n})\cap H^{0,5}(\mathbb{R}^{n})$
		and $\|\varphi\|_{5,0}+\|\varphi\|_{0,5}=\epsilon'<\epsilon$, where
		$\epsilon$ is sufficiently small . Then there exists a unique global
		solution $\varphi_{t}$ the Hartree equation \eqref{eq:Hartree},(with
		\[
		V(x)=\lambda|x|^{-\delta}
		\]
		for $0<\delta<1$, such that 
		\[
		\|\varphi_{t}\|_{\infty}\leq C\epsilon'(1+|t|)^{-3/2}.
		\]
	\end{lem}
	
	\begin{proof}
		See Theorem 1.1 of \cite{Hayashi01}. If we put $n=3$ and $p=\infty$ for our discussion, we get the result.
	\end{proof}
	
	\begin{lem}[Hayashi and Ozawa 87']
		\label{prop:Hayashi98}We assume that $\varphi\in H^{2,0}(\mathbb{R}^{n})\cap H^{0,2}(\mathbb{R}^{n})$.
		Then there exists a unique global solution $\varphi_{t}$ of the Hartree
		type equation \eqref{eq:Hartree}, with 
		\[
		V(x)=|x|^{-1}
		\]
		Then,
		\[
		\|\varphi_{t}\|_{\infty}\leq C(1+|t|)^{-1/2}.
		\]
	\end{lem}
	
	\begin{proof}
		See Theorem 1.1 \cite{HayashiOzawa87}.
	\end{proof}
	
	Notice that Lemma \ref{prop:HayashiNaumkin98} and \ref{prop:HayashiNaumkin01} were proven under the condition of small initial data.
	We will interpret
	(or convert) this result into the case of generic initial data with
	weak interaction. The strategy is the following:
	
	We substitute $\varphi$ with $\widetilde{\varphi}/(\epsilon'/M)$
	for suitable constant $M>0$. Then $\widetilde{\varphi}$ solves the partial differential equation
	\[
	\mathrm{i}\partial_{t}\widetilde{\varphi}_{t}=-\Delta\widetilde{\varphi}_{t}+((\epsilon'/M)^{2}V*|\widetilde{\varphi}_{t}|^{2})\widetilde{\varphi}_{t}
	\]
	with initial data $\|\widetilde{\varphi}\|_{\gamma,0}+\|\widetilde{\varphi}\|_{0,\gamma}=M$.
	Now, letting $\widetilde{\lambda}:=\lambda\epsilon'^{2}/M^{2}$,
	\[
	(\epsilon'/(M))^{2}V=\frac{\lambda\epsilon'^{2}}{M^{2}}\frac{e^{-\mu|x|}}{|x|^{\gamma}}=\widetilde{\lambda}\frac{e^{-\mu|x|}}{|x|^{\gamma}}.
	\]
	Note that $\epsilon'=\epsilon'(\lambda)$ was small enough and $M>0$
	was arbitrarily chosen. Hence, we have new Hatree equation
	\[
	\mathrm{i}\partial_{t}\widetilde{\varphi}_{t}=-\Delta\widetilde{\varphi}_{t}+(\widetilde{\lambda}\frac{e^{-\mu|x|}}{|x|^{\gamma}}*|\widetilde{\varphi}_{t}|^{2})\widetilde{\varphi}_{t}
	\]
	with small interaction constant $\widetilde{\lambda}$ such that $|\widetilde{\lambda}|\leq\lambda_c=\lambda_c(\mu,\gamma,M)$. Therefore,
	using this `interpretation', we have Proposition \ref{prop:timedecay}
	and Proposition \ref{prop:timedecay_defocus}.

	\subsection{Time decay estimates of the Hartree equation for Yukawa type interaction}
	
	In this section, we provide decay estimates for Yukawa type interaction potential.
	Since the proofs will closely follow \cite{HayashiNaumkin98} and \cite{HayashiOzawa87}, we only provide
	the sketch of proofs. For time decay estimates, heuristically, the main difficulty stems from  
	attractive, long-range interaction potential; if the range of the interaction is short enough, then `far sides' of wave
	function would not interact with each other. Hence, if there is a time decay estimate
	for Coulomb interaction, one can also expect that there is a similar bound for Yukawa type interaction. 
	Even though the explanation here is rather heuristic, this can be made rigorous as in the following lemmas, whose proofs are based on fixed point arguments.
	
	\begin{lem}
		\label{prop:Yukawa}We assume that $\varphi\in H^{S,0}(\mathbb{R}^{n})\cap H^{0,S}(\mathbb{R}^{n})$
		and $\|\varphi\|_{S,0}+\|\varphi\|_{0,S}=\epsilon'<\epsilon$,where
		$\epsilon$ is sufficiently small and $3/2<S<5/3$. Then there exists
		a unique global solution $\varphi_{t}$ of the Hartree equation
		\eqref{eq:Hartree}, with 
		\[
		V(x)=\frac{\lambda e^{-\mu|x|}}{|x|^{\gamma}}
		\]
		for $\mu>0$ and $0<\gamma<3/2$, such that 
		\[
		\varphi_{t}\in C(\mathbb{R},H^{S,0}\cap H^{0,S})
		\]
		and
		\[
		\|\varphi_{t}\|_{\infty}\leq C\epsilon'(1+|t|)^{-3/2}.
		\]
	\end{lem}
	
	\begin{proof}[Idea of proof]
		Because the proof in \cite{HayashiNaumkin98} relies on the fact
		that
		\[
		\||x|^{-\gamma}*|u|\|_{L^{p}(\mathbb{R}^{n})}<\infty
		\]
		and
		\[
		\|(-t^{2}\Delta)^{s/2}|x|^{-\gamma}*|u|\|_{L^{p}(\mathbb{R}^{n})}<\infty
		\]
		for some $t>0$, $0<s<1$, $1\leq\gamma<3/2$, and $n\in\mathbb{Z}$,
		we have
		\[
		\|\exp(-\mu|x|)|x|^{-\gamma}*|u|\|_{L^{p}(\mathbb{R}^{n})}<\infty
		\]
		and
		\begin{align*}
		&\|(-t^{2}\Delta)^{s/2}\exp(-\mu|x|)|x|^{-\gamma}*|u|\|_{L^{p}(\mathbb{R}^{n})}\\ &\quad\leq\|(-t^{2}\Delta)^{s/2}\exp(-\mu|x|)|x|^{-\gamma}*|u|\|_{L^{p}(\mathbb{R}^{n})}<\infty.
		\end{align*}
	\end{proof}
	
	\begin{lem}
		\label{prop:Yukawa_HayashiOzawa87}We assume that $\varphi\in H^{2,0}(\mathbb{R}^{n})\cap H^{0,2}(\mathbb{R}^{n})$.
		Then there exists a unique global solution $\varphi_{t}$ of the Hartree
		type equation \eqref{eq:Hartree}, with 
		\[
		V(x)= \frac{\lambda e^{-\mu|x|}}{|x|^\gamma}
		\]
		Then,
		\[
		\|\varphi_{t}\|_{\infty}\leq C(1+|t|)^{-1/2}.
		\]
	\end{lem}
	
	\begin{proof}[Idea of proof]
		Noting that $e^{-\mu|x|}|x|^{-\gamma}<C|x|^{-1}$ for some $C=C(\mu,\gamma)$. We follow the proof of \cite{HayashiOzawa87}.
	\end{proof}
	
	\subsection{On the time dependence of $\int_{0}^{t}\mathrm{d}s\,\|V(x-\cdot)\varphi_{s}\|_{2}$}
	
	We are going to prepare for Section \ref{sec:comparison}.
	Proposition \ref{prop:key_estimate} below is the key lemma to improve the time dependence of the
	Lemma presented in Section \ref{sec:comparison}. The proof of Proposition \ref{prop:key_estimate} is based on the following
	two lemmas.
	\begin{lem}[Boundedness of $H^{1}$-norm of $\varphi_{t}$]
		\label{lem:bdd_H1} For the solution $\varphi_{t}$ of the Hartree equation
		\eqref{eq:Hartree} for $V\in L^{2}+L^{\infty}$, $\varphi_{0}\in H^{1}(\mathbb{R}^{3})$
		for the Hartree equation then there exist constant $C$ depending only
		on $\varphi_{0}$ and $V$ such that
		\[
		\|\varphi_{t}\|_{H^{1}(\mathbb{R}^{3})}\leq C.
		\]
	\end{lem}
	
	\begin{proof}
		See Lemma 2.1 of \cite{Chen2018}.
	\end{proof}
	
	\begin{lem}[Strichartz estimate for $V\in L^2$]\label{lem:decaying}
		Suppose that $V\in L^2(\mathbb{R}^3)$. Let $\varphi_{t}$ be the solution of the Hartree equation \eqref{eq:Hartree} with initial data $\varphi_{0}=\varphi\in H^{1}(\mathbb{R}^3)$, then there exists a constant $C$, depending only on $\left\Vert \varphi\right\Vert _{H^{1}}$ and $\|V\|_{L^2}$, such that
		\[\|\varphi_t\|_{L^2((0,T),L^\infty)} \leq { C \sqrt{1+ {{T}} }}.\]
	\end{lem}
	
	\begin{proof}
		We closely follow \cite[Theorem 2.3.3]{cazenave2003semilinear} for the proof of the lemma. The result for $V\in L^{2}+L^{\infty}$ is in the proof of Lemma \ref{lem:bdd_H1} and here we remove terms for $L^\infty$ part of $V$. From the Sobolev inequality and the Strichartz's estimate,
		we have
		\begin{equation}
		\begin{aligned}
		\|\varphi_t\|_{L^2((0,T),L^\infty)} &\leq C \|\varphi_t\|_{L^2((0,T),W^{1,6})}\\
		&\leq C\|\varphi_0\|_{H^1} + C\|(V*|\varphi_t|^2)\varphi_t\|_{L^2((0,T),W^{1,6/5})}.
		\end{aligned}\label{eq:Strischartz_div}
		\end{equation}
		From the definition of the Sobolev norm,
		\begin{align}	
		&\|(V*|\varphi_t|^2)\varphi_t\|_{L^2((0,T),W^{1,6/5})} \label{eq:(V*phi^2)phi} \\
		&\quad\leq C \|(V*|\varphi_t|^2)\varphi_t\|_{L^2((0,T),{L^{6/5}})} + C \|\nabla ((V*|\varphi_t|^2)\varphi_t)\|_{L^2((0,T),{L^{6/5}})}. \notag
		\end{align}
		We first focus on the spacial integral; integration with respect to the time variable $t$ will be considered later. In the first term in the right-hand side of \eqref{eq:(V*phi^2)phi}, the integrand of the spatial integral is bounded by
		\begin{equation}
		\begin{aligned}
		\|(V*|\varphi_t|^2)\varphi_t\|_{L^{6/5}}
		&\leq \|V*|\varphi_t|^2\|_{L^3}\|\varphi_t\|_{L^2} \leq \|V\|_{L^2}\||\varphi_t|^2\|_{L^{6/5}}\|\varphi_t\|_{L^2}  \\
		&\leq  \|V\|_{L^2}\|\varphi_t\|_{L^{12/5}}^2\|\varphi_t\|_{L^2}
		\leq  \|V\|_{L^2}\|\varphi_t\|_{L^{2}}^{5/2}\|\varphi_t\|_{L^6}^{1/2}\\
		&\leq \|V\|_{L^2}\|\varphi_t\|_{L^2}^{5/2}\|\varphi_t\|_{H^1}^{1/2},
		\end{aligned}
		\end{equation}
		where we used H\"older's inequality, Young's inequality, and Riesz-Thorin Theorem. Similarly, we decompose the integrand of the second term in the right-hand side \eqref{eq:(V*phi^2)phi} into two parts and find that
		\begin{align*}
		&\|\nabla ((V*|\varphi_t|^2)\varphi_t)\|_{L^2((0,T),{L^{6/5}})}\\
		&\quad\leq \| (V*(\nabla|\varphi_t|^2))\varphi_t\|_{L^2((0,T),{L^{6/5}})} +\| (V*|\varphi_t|^2)(\nabla\varphi_t)\|_{L^2((0,T),{L^{6/5}})}.
		\end{align*}
		We again apply H\"older's inequality, Young's inequality, and Riesz-Thorin Theorem to get
		\begin{align*}
		\| V*(\nabla |\varphi_t|^2 ) \varphi_t\|_{L^{6/5}} &\leq \| V*(\nabla|\varphi_t|^2)\|_{L^{3}}\|\varphi_t\|_{L^2}
		\leq  C \|V\|_{L^2}\|\overline{\varphi_t}\nabla\varphi_t\|_{L^{6/5}}\|\varphi_t\|_{L^2}\\
		&\leq  C \|V\|_{L^2}\|\varphi_t\|_{L^{3}}\|\nabla\varphi_t\|_{L^{2}}\|\varphi_t\|_{L^2}
		\leq C \|V\|_{L^2}\|\varphi_t\|_{L^{2}}^{3/2}\|\varphi_t\|_{H^1}^{3/2}, \\
		\| (V*|\varphi_t|^2)(\nabla\varphi)\|_{L^{6/5}}
		&\leq \|V*|\varphi_t|^2\|_{L^{3}}\|\nabla\varphi_t\|_{L^2}
		\leq \|V\|_{L^2}\|\varphi_t\|_{L^2}^{3/2}\|\varphi_t\|_{H^1}^{3/2
		}.
		\end{align*}
		Thus, after taking $L^2$-norm according to \eqref{eq:Strischartz_div} with respect to the time variable $t$, with the mass conservation $\|\varphi_t\|_{L^2}=1$ and Lemma \ref{lem:bdd_H1}, we conclude that
		\[
		\|\varphi_t\|_{L^2((0,T),L^\infty)} \leq  C \sqrt{1 + {{T}}}.
		\]
	\end{proof}
	
	\begin{prop}[Key estimate]
		\label{prop:key_estimate}Suppose that $\varphi_{s}$ a solution
		of \eqref{eq:Hartree} with initial data $\varphi$ satisfies Assumption \ref{ass:cases}. We have
		\begin{equation}
		\int_{0}^{t}\mathrm{d}s\,\|V(x-\cdot)\varphi_{s}\|_{2}\leq C(t).\label{eq:key_estimate}
		\end{equation}
		where $C(t)=C(\varphi_{0},V,t)$ is depends only on initial data
		$\varphi_{0}$, interaction potential $V$ and time $t$, given in Table \ref{table:RoC_t}.
	\end{prop}
	\begin{rem}
		Strichartz estimate was used to obtain
		\[
		\int_{0}^{t}\mathrm{d}s\sup_{x}\|V(\cdot-x)\varphi_{s}\|_{2}\leq C(1+t)^{3/2}
		\]
		in \cite{Chen2018}.
		Here we use Proposition \ref{prop:key_estimate} so that
		\[
		\int_{0}^{t}\mathrm{d}s\sup_{x}\|V(\cdot-x)\varphi_{s}\|_{2}\leq C(t).
		\]
	\end{rem}
	\begin{proof}[Proof of Proposition \ref{prop:key_estimate}]
		
		Throughout this proof, (i) for $|\lambda|<\lambda_c$, we use the time decay estimate to prove a sub-exponential bound in time, and (ii) for $|\lambda|>\lambda_c$, we prove an exponential (or slightly bigger) bound in time without time decay estimate.
		
		For Coulomb cases, we consider the following: For a fixed $x\in\mathbb{R}^{3}$, let $B_{r}=\{y\in\mathbb{R}^{3}:|x-y|\leq r\}$
		be the ball centered at $x$ with radius $r$.
		By H\"older inequality, the fact that $|x-y|^{-2\gamma}<1$
		for $y\in B_{1}^{c}$, Sobolev embedding, and Lemma \ref{lem:bdd_H1},
		we have
		\begin{align}
		\frac{1}{\lambda^{2}}\|V(x-\cdot)\varphi_{s}\|_{2}^{2} & =\int\mathrm{d}y\,\frac{|\varphi_{s}(y)|^{2}}{|x-y|^{2\gamma}}\nonumber  =\int_{B_{f(s)}}\mathrm{d}y\,\frac{|\varphi_{s}(y)|^{2}}{|x-y|^{2\gamma}}+\int_{B_{f(s)}^{c}}\mathrm{d}y\,\frac{|\varphi_{s}(y)|^{2}}{|x-y|^{2\gamma}}\nonumber \\
		& \leq C\|\varphi_{s}\|_{\infty}^{2}\left(\int_{B_{f(s)}}\mathrm{d}y\,\frac{1}{|x-y|^{2\gamma}}\right)+C\left(f(s)\right)^{-2\gamma}\|\varphi_{s}\|_{2}^{2}\label{eq:V_generic_bdd_decay}
		\end{align}
		for 
		a positive valued function $f(s)$ with arbitrary $s>0$, which will be determined later.
		
		Note that
		\[
		\int_{B_{f(s)}}\mathrm{d}y\,\frac{1}{|x-y|^{2\gamma}}=4\pi\int_{0}^{f(s)}r^{2-2\gamma}\mathrm{d}r=\frac{4\pi}{3-2\gamma}\left(f(s)\right)^{3-2\gamma}
		\]
		implies, by time decay estimate, that
		\begin{align*}
		\frac{1}{\lambda^{2}}\|V(x-\cdot)\varphi_{s}\|_{2}^{2} & \leq C\|\varphi_{s}\|_{\infty}^{2}\left(f(s)\right)^{3-2\gamma}+C\left(f(s)\right)^{-2\gamma}\\
		&\leq C(1+s)^{-3}\left(f(s)\right)^{3-2\gamma}+C\left(f(s)\right)^{-2\gamma}.
		\end{align*}
		By letting $f(s)=1+s$, we get
		\begin{equation}
		\|V(x-\cdot)\varphi_{s}\|_{2}^{2}\leq C(1+s)^{-2\gamma}.\label{eq:small_lambda}
		\end{equation}
		
		\noindent \textbf{Case 1.} $V(x)=\lambda|x|^{-\gamma}$ with $0<\gamma\leq1$
		and $\lambda\in\mathbb{R}$.
		
		From H\"older inequality and Hardy inequality, we get
		\begin{align*}
		\frac{1}{\lambda^{2}}\|V(x-\cdot)\varphi_{s}\|^2_{2} & =\int\mathrm{d}y\,|x-y|^{-2\gamma}|\varphi_{s}(y)|^{2} \\
		&=\int\mathrm{d}y\,|x-y|^{-2\gamma}|\varphi_{s}(y)|^{2\gamma}\cdot|\varphi_{s}(y)|^{2-2\gamma}\\
		& \leq\left(\int\mathrm{d}y\,|x-y|^{-2}|\varphi_{s}(y)|^{2}\right)^{\gamma}\left(\int\mathrm{d}y\,|\varphi_{s}(y)|^{2}\right)^{1-\gamma}\\
		&\leq C\|\varphi_{s}\|_{H^{1}}^{2\gamma}\|\varphi_{s}\|_{L^{2}}^{2-2\gamma}\\
		& \leq C\|\varphi_{s}\|_{H^{1}}^{2} \leq C.
		\end{align*}
		
		\noindent \textbf{Case 2.} $V(x)=\lambda|x|^{-\gamma}$ with $0<\gamma<1$
		and $|\lambda|\leq\lambda_c$.
		
		We have from \eqref{eq:small_lambda} that 
		\[
		\int_{0}^{t}\mathrm{d}s\,\|V(x-\cdot)\varphi_{s}\|_{2}\leq\int_{0}^{t}\mathrm{d}s\,(1+s)^{-\gamma}\leq\frac{1}{1-\gamma}(1+t)^{1-\gamma}.
		\]
		Then
		\[
		\exp\left(\int_{0}^{t}\mathrm{d}s\,\|V(x-\cdot)\varphi_{s}\|_{2}\right)\leq\exp\left(K(1+t)^{1-\gamma}\right).
		\]
		
		\noindent \textbf{Case 3.1.} $V(x)=\lambda|x|^{-\gamma}$ with $\gamma=1$
		and $|\lambda|\leq\lambda_c$.
		
		From Kato's inequality and \eqref{eq:small_lambda},
		
		\[
		\int_{0}^{t}\mathrm{d}s\,\|V(x-\cdot)\varphi_{s}\|_{2}\leq\int_{0}^{t}\mathrm{d}s\,2\sqrt{\pi}|\lambda|(1+s)^{-1}\leq2\sqrt{\pi}||\lambda|\log(1+t).
		\]
		Then
		\[
		\exp\left(\int_{0}^{t}\mathrm{d}s\,\|V(x-\cdot)\varphi_{s}\|_{2}\right)\leq(1+t)^{2\sqrt{\pi}||\lambda|}.
		\]
		
		\noindent \textbf{Case 3.2.} $V(x)=\lambda|x|^{-\gamma}$ with $1<\gamma<3/2$
		and $|\lambda|\leq\lambda_c$.
		
		From \eqref{eq:small_lambda},
		\[
		\int_{0}^{t}\mathrm{d}s\,\|V(x-\cdot)\varphi_{s}\|_{2}\leq C\int_{0}^{t}\mathrm{d}s\,(1+s)^{-\gamma}\leq \frac{C}{\gamma -1}.
		\]
		Thus
		\[
		\exp\left(\int_{0}^{t}\mathrm{d}s\,\|V(x-\cdot)\varphi_{s}\|_{2}\right)\leq C.
		\]
		
		\noindent \textbf{Case 3.3.} $V(x)=\lambda|x|^{-\gamma}$ with $0<\gamma<3/2$
		and $\lambda\in\mathbb{R}$. 
		
		Let $2\gamma/3\leq\alpha=\alpha(\gamma)<1$
		so that
		\begin{align}
		\frac{1}{\lambda^{2}}\|V(x-\cdot)\varphi_{s}\|^2_{2} & =\int\mathrm{d}y\,\frac{|\varphi_{s}(y)|^{2}}{|x-y|^{2\gamma}}\nonumber \\
		& =\int_{B_{1}}\mathrm{d}y\,\frac{|\varphi_{s}(y)|^{2}}{|x-y|^{2\gamma}}+\int_{B_{1}^{c}}\mathrm{d}y\,\frac{|\varphi_{s}(y)|^{2}}{|x-y|^{2\gamma}}.\label{eq:coulomb_gamma_big}
		\end{align}
		Note that by H\"older inequality with a pair $(\frac{3\alpha}{2\gamma},\frac{3\alpha}{3\alpha-2\gamma})$,
		we have
		\begin{align*}
		\int_{B_{1}}\mathrm{d}y\,\frac{|\varphi_{s}(y)|^{2}}{|x-y|^{2\gamma}} & \leq\left\Vert |x-\cdot|^{{-2\gamma}}\right\Vert _{\frac{3\alpha}{2\gamma}}\||\varphi_{s}|^{2}\|_{\frac{3\alpha}{3\alpha-2\gamma}} \leq\left(\int_{B_{1}}\mathrm{d}y\,|x-\cdot|^{-3\alpha}\right)^{\frac{2\gamma}{3\alpha}}\|\varphi_{s}\|_{\frac{6\alpha}{3\alpha-2\gamma}}^{2}.
		\end{align*}
		Since $\alpha<1$, the first factor $\left(\int_{B_{1}}\mathrm{d}y\,|x-\cdot|^{-3\alpha}\right)^{2\gamma/3\alpha}=:C_{\alpha}<\infty$.
		By Riesz--Thorin theorem
		\[
		\|\varphi_{s}\|_{\frac{6\alpha}{3\alpha-2\gamma}}\leq\|\varphi_{s}\|_{6}^{\frac{3\alpha-2\gamma}{\alpha}}\|\varphi_{s}\|_{\infty}^{\frac{2\gamma-2\alpha}{\alpha}}.
		\]
		Thus
		\[
		\int_{B_{1}}\mathrm{d}y\,\frac{|\varphi_{s}(y)|^{2}}{|x-y|^{2\gamma}}\leq C \|\varphi_{s}\|_{\infty}^{4(\gamma-\alpha)/\alpha}.
		\]
		Next, we bound the second term of \eqref{eq:coulomb_gamma_big}
		using that $|x-y|^{-2\gamma}\leq1$ for $y\in B_{1}^{c}$ so that
		\[
		\int_{B_{1}^{c}}\mathrm{d}y\,\frac{|\varphi_{s}(y)|^{2}}{|x-y|^{2\gamma}}\leq\int_{B_{1}^{c}}\mathrm{d}y\,|\varphi_{s}(y)|^{2}\leq\|\varphi_{s}\|_{2}^{2}\leq C.
		\]
		Hence,
		\begin{align*}
		\frac{1}{\lambda^{2}}\|V(x-\cdot)\varphi_{s}\|^{2}_{2} & \leq C\|\varphi_{s}\|_{\infty}^{4(\gamma-\alpha)/\alpha}+C.
		\end{align*}
		Now we have, using H\"older inequality in time and Strichartz estimate,
		\begin{align*}
		\int_{0}^{t}\mathrm{d}s\,\|V(x-\cdot)\varphi_{s}\|_{2} & \leq C\int_{0}^{t}\mathrm{d}s\,(\|\varphi_{s}\|_{\infty}^{2(\gamma-\alpha)/\alpha}+1)\\
		&\leq C\left(\int_{0}^{t}\mathrm{d}s\right)^{(2\alpha-\gamma)/\alpha}\left(\int_{0}^{t}\mathrm{d}s\,\|\varphi_{s}\|_{\infty}^{2}\right)^{2(\gamma-\alpha)/2\alpha}+Ct\\
		& \leq C(1+t)^{\frac{2\alpha-\gamma}{\alpha}+\frac{2\gamma-2\alpha}{\alpha}}+Ct
		\leq \max\{C(1+t)^{\gamma/\alpha},Ce^t\}
		\end{align*}
		for any $\alpha\in[2\gamma/3,1)$.
		
		\noindent \textbf{Case 3.4.} $V(x)=\lambda\exp(-\mu|x|)|x|^{-\gamma}$
		with $0<\gamma<3/2$, $\mu>0$, and $\lambda\in\mathbb{R}$.
		
		Note that then $V\in L^2$. By H\"older inequality and Sobolev embedding,
		\begin{align*}
		\frac{1}{\lambda^{2}}\|V(x-\cdot)\varphi_{s}\|^{2}_{2} & =\int\mathrm{d}y\,\frac{e^{-2\mu|x-y|}}{|x-y|^{2\gamma}}|\varphi_{s}(y)|^{2}\\ &\leq\|\varphi_{s}(y)\|_{\infty}^{2}\int\mathrm{d}y\,\frac{e^{-2\mu|x-y|}}{|x-y|^{2\gamma}} \leq C\|\varphi_{s}\|_{\infty}^{2}.
		\end{align*}
		Then by Cauchy--Schwarz inequality and Lemma \ref{lem:decaying} (Strichartz estimate), we get
		\begin{align*}
		\int_{0}^{t}\mathrm{d}s\,\|V(x-\cdot)\varphi_{s}\|_{2} & \leq C\int_{0}^{t}\mathrm{d}s\,\|\varphi_{s}\|_{\infty}
		\leq C\left(\int_{0}^{t}\mathrm{d}s\right)^{1/2}\\
		&\left(\int_{0}^{t}\mathrm{d}s\|\varphi_{s}\|_{\infty}^2\right)^{1/2}
		\leq C (1+t).
		\end{align*}
		
		\noindent \textbf{Case 4.} $V(x)=\lambda\exp(-\mu|x|)|x|^{-\gamma}$
		with $\mu>0$, $0<\gamma<3/2$, and $|\lambda|\leq\lambda_c$.
		
		Using H\"older inequality,
		\begin{align*}
		\frac{1}{\lambda^{2}}\|V(x-\cdot)\varphi_{s}\|^{2}_{2} & =\int\mathrm{d}y\,\frac{e^{-2\mu|x-y|}|\varphi_{s}(y)|^{2}}{|x-y|^{2\gamma}} \\&\leq C\|\varphi_{s}\|_{\infty}^{2}\int\mathrm{d}y\,\frac{e^{-2\mu|x-y|}}{|x-y|^{2\gamma}} \leq C(1+s)^{-3},
		\end{align*}
		hence we get a time independent bound
		\[
		\int_{0}^{t}\mathrm{d}s\,\|V(x-\cdot)\varphi_{s}\|_{2}\leq\int_{0}^{t}\mathrm{d}s\,C(1+s)^{-3/2}\leq C.
		\]
		
		\noindent \textbf{Case 5.} $V(x)=\lambda|x|^{-\gamma}$ with $1\leq\gamma<3/2$
		and $\lambda>0$.
		
		From \eqref{eq:V_generic_bdd_decay},
		\begin{align*}
		\frac{1}{\lambda^{2}}\|V(x-\cdot)\varphi_{s}\|^{2}_{2} & \leq C\|\varphi_{s}\|_{\infty}^{2}\left(f(s)\right)^{3-2\gamma}+C\left(f(s)\right)^{-2\gamma} \\
		&\leq C(1+s)^{-1}\left(f(s)\right)^{3-2\gamma}+C\left(f(s)\right)^{-2\gamma}.
		\end{align*}
		By putting $f(s)={(1+s)^{1/3}}$, we find that
		\[
		\|V(x-\cdot)\varphi_{s}\|^{2}_{2}\leq C(1+s)^{-2\gamma/3},
		\]
		hence
		\[
		\exp\left(\int_{0}^{t}\mathrm{d}s\,\|V(x-\cdot)\varphi_{s}\|_{2}\right)\leq C\exp\left(K(1+t)^{1-2\gamma/3}\right).
		\]
		
		\noindent \textbf{Case 6.} $V(x)=\lambda\exp(-\mu|x|)|x|^{-\gamma}$
		with $0<\gamma<3/2$, $\mu>0$, and $\lambda>0$.
		
		Similarly to \eqref{eq:V_generic_bdd_decay},
		\begin{align*}
		\frac{1}{\lambda^{2}}\|V(x-\cdot)\varphi_{s}\|^{2}_{2}
		& \leq C\|\varphi_{s}\|_{\infty}^{2}\left(\int_{B_{f(s)}}\mathrm{d}y\,\frac{e^{-2\mu|x-y|}}{|x-y|^{2\gamma}}\right)+C\left(f(s)\right)^{-2\gamma}\|\varphi_{s}\|^{2}_{2}\\
		&\leq C\|\varphi_s\|_\infty^2+C\frac{e^{-2f(s)}}{f(s)^{2\gamma}}.
		\end{align*}
		Letting $f(s)={1+s}$, we get
		\[
		\frac{1}{\lambda^{2}}\|V(x-\cdot)\varphi_{s}\|^{2}_{2}\leq C(1+s)^{-1},
		\]
		which implies that
		\begin{align*}
		\int_{0}^{t}\mathrm{d}s\,\|V(x-\cdot)\varphi_{s}\|_{2} & \leq C\int_{0}^{t}\mathrm{d}s (1+s)^{-1}
		\leq C \log(1+t).
		\end{align*}
		This completes the proof.
	\end{proof}
	
	\begin{rem}
		In the table below, we summarize the cases considered in the proof of Proposition \ref{prop:key_estimate}.
		\begin{center}
			\begin{table}[h]
				\caption{The cases in the proof of Proposition \ref{prop:key_estimate}}
				\label{table:RoC_t2}
				\centering{}
				\begin{tabular}{|c|c|c|c|c|}
					\hline 
					\multirow{2}{*}{} & \multicolumn{3}{c|}{$V(x)=\lambda|x|^{-\gamma}$} & $V(x)=\lambda\exp(-\mu|x|)|x|^{-\gamma}$, $\mu>0$\tabularnewline
					\cline{2-5} 
					& $0<\gamma<1$ & $\gamma=1$ & $1<\gamma<3/2$ & $0<\gamma<3/2$\tabularnewline
					\hline 
					$\lambda>\lambda_c$ & Case1 & \multicolumn{2}{c|}{Case 5} & Case 6\tabularnewline
					\cline{1-5} 
					$|\lambda|<\lambda_c$ & Case 2 & Case 3.1 & Case 3.2 & Case 4\tabularnewline
					\cline{1-5} 
					$\lambda<-\lambda_c$ & \multicolumn{2}{c|}{Case 1}  & Case 3.3 & Case 3.4\tabularnewline
					\hline 
				\end{tabular}
			\end{table}
			\par\end{center}
	\end{rem}
	
	\section{Fock space formalism}
	
	This section is devoted to explain Fock space formalism for studying the dynamics of the system of $N$-Bosons. We consider Bosonic Fock space as in \cite{Chen2011,Lee2013,Rodnianski2009}.
	The Bosonic Fock space is a Hilbert space defined by
	\[
	\mathcal{F}=\bigoplus_{n\geq0}L^{2}\left(\mathbb{R}^{3},\mathrm{d}x\right)^{\otimes_{s}n}=\mathbb{C}\oplus\bigoplus_{n\geq1}L_{s}^{2}\left(\mathbb{R}^{3n},\mathrm{d}x_{1},\dots,\mathrm{d}x_{n}\right),
	\]
	where $L_{s}^{2}=L_{s}^{2}(\mathbb{R}^{3n},\mathrm{d}x_{1},\dots,\mathrm{d}x_{n})$
	is a subspace of $L^{2}(\mathbb{R}^{3n},\mathrm{d}x_{1},\dots,\mathrm{d}x_{n})$
	that is the space of all functions symmetric under any permutation
	of $x_{1},x_{2},\dots,x_{n}$. It is convenient to let $L_{s}^{2}(\mathbb{R}^{3})^{\otimes0}=\mathbb{C}$.
	An element $\psi\in\mathcal{F}$ can be understood as a sequence $\psi=\{\psi^{\left(n\right)}\}_{n\geq0}$
	of $n$-particle wave functions $\psi^{\left(n\right)}\in L_{s}^{2}(\mathbb{R}^{3n})$
	or as a vector in a countable dimensional vector space such that each
	$n$-th component is a function $\psi^{\left(n\right)}\in L_{s}^{2}(\mathbb{R}^{3n})$.
	The inner product on $\mathcal{F}$ is defined by 
	\[
	\begin{aligned}\langle\psi_{1},\psi_{2}\rangle_{\mathcal{F}} & =\sum_{n\geq0}\langle\psi_{1}^{\left(n\right)},\psi_{2}^{\left(n\right)}\rangle_{L^{2}\left(\mathbb{R}^{3n}\right)}\\
	& =\overline{\psi_{1}^{(0)}}\psi_{2}^{(0)}+\sum_{n\geq0}\int\mathrm{d}x_{1}\dots\mathrm{d}x_{n}\,\overline{\psi_{1}^{\left(n\right)}\left(x_{1},\dots,x_{n}\right)}\,\psi_{2}^{\left(n\right)}\left(x_{1},{\dots},x_{n}\right).
	\end{aligned}
	\]
	We denote $\|\psi\|_{\mathcal{F}}=\langle\psi,\psi\rangle_\mathcal{F}^{1/2}$.
	The vector $\Omega:=\left\{ 1,0,0,{\dots}\right\} \in\mathcal{F}$
	is called the vacuum. Note that an element $\psi\in\mathcal{F}$ is
	denoting a many-body quantum state which can have uncertainty of the number
	of particles of the quantum system. Because of that one can think
	of generation or annihilation of a particle. For $f\in L^{2}(\mathbb{R}^{3})$,
	we define the creation operator $a^{*}(f)$ and the annihilation operator
	$a(f)$ on $\mathcal{F}$ by 
	\begin{equation}
	\left(a^{*}\left(f\right)\psi\right)^{\left(n\right)}(x_{1},\dots,x_{n})=\frac{1}{\sqrt{n}}\sum_{j=1}^{n}f(x_{j})\psi^{\left(n-1\right)}(x_{1},\dots,x_{j-1},x_{j+1},\dots,x_{n})\label{eq:creation}
	\end{equation}
	and 
	\begin{equation}
	\left(a\left(f\right)\psi\right)^{\left(n\right)}\left(x_{1},\dots,x_{n}\right)=\sqrt{n+1}\int\mathrm{d}x\overline{f\left(x\right)}\psi^{\left(n+1\right)}\left(x,x_{1},\dots,x_{n}\right),\label{eq:annihilation}
	\end{equation}
	each of which denotes the creation or annihilation of a particle having wave function $f$.
	By definition, the creation operator $a^{*}(f)$ is the adjoint of
	the annihilation operator of $a(f)$, and in particular, $a^{*}(f)$
	and $a(f)$ are not self-adjoint. We will use the self-adjoint operator
	$\phi\left(f\right)$ defined as 
	\[
	\phi(f)=a^{*}(f)+a(f).
	\]
	Let $a_{x}^{*}$ and $a_{x}$ operator-valued distributions such that
	\[
	a^{*}(f)=\int\mathrm{d}x\,f\left(x\right)a_{x}^{*},\qquad a(f)=\int\mathrm{d}x\,\overline{f\left(x\right)}a_{x}
	\]
	for any $f\in L^{2}(\mathbb{R}^{3})$. 
	For each non-negative
	integer $n$, we introduce the projection operator onto the $n$-particle
	sector of the Fock space, for $\psi=(\psi^{(0)},\psi^{(1)},\dots)\in\mathcal{F}$,
	\begin{equation}
	P_{n}(\psi):=(0,0,\dots,0,\psi^{(n)},0,\dots).\label{eq:projection}
	\end{equation}
	For simplicity,
	with slight abuse of notation, we will use $\psi^{(n)}$ to denote
	$P_{n}\psi$. The will use number operator $\mathcal{N}$ which counts the expected number of particles of a vector in $\mathcal{F}$ and is defined by 
	\begin{equation}
	\mathcal{N}=\int\mathrm{d}x\,a_{x}^{*}a_{x}.\label{eq:number operator}
	\end{equation}
	Note that $\mathcal{N}$ satisfies that $\left(\mathcal{N}\psi\right)^{\left(n\right)}=n\psi^{\left(n\right)}$.
	Let $J$ be an operator defined on the one-particle sector
	$L^{2}\left(\mathbb{R}^{3},\mathrm{d}x\right)$, then 
	we extend this operator into Fock space by $d\Gamma\left(J\right)$, which is called
	its second quantization and whose action
	on the $n$-particle sector is given by 
	\[
	\left(d\Gamma\left(J\right)\psi\right)^{\left(n\right)}=\sum_{j=1}^{n}J_{j}\psi^{\left(n\right)}
	\]
	where $J_{j}=1\otimes{\dots}\otimes J\otimes{\dots}\otimes1$ is the
	operator $J$ acting on the $j$-th variable only. 
	With a kernel
	$J\left(x;y\right)$ of the operator $J$, the second quantization
	$d\Gamma\left(J\right)$ can be also be written as 
	\[
	d\Gamma\left(J\right)=\int\mathrm{d}x\mathrm{d}y\,J\left(x;y\right)a_{x}^{*}a_{y}.
	\]
	
	The following lemma shows that the annihilation operator
	and the creation operator can be bounded roughly $\mathcal{N}^{1/2}$ or $(\mathcal{N}+1)^{1/2}$. Moreover, it gives a bound of the second quantization operators.
	
	\begin{lem}[Lemma 2.1 in \cite{Chen2011a}]
		For $\alpha>0$, let $D(\mathcal{N}^{\alpha})=\{\psi\in\mathcal{F}:\sum_{n\geq1}n^{2\alpha}\|\psi^{(n)}\|^{2}_{2}<\infty\}$
		denote the domain of the operator $\mathcal{N}^{\alpha}$. For any
		$f\in L^{2}(\mathbb{R}^{3},dx)$ and any $\psi\in D(\mathcal{N}^{1/2})$,
		we have 
		\begin{equation}
		\begin{split}\|a(f)\psi\|_{\mathcal{F}} & \leq\|f\|_{2}\,\|\mathcal{N}^{1/2}\psi\|_{\mathcal{F}},\\
		\|a^{*}(f)\psi\|_{\mathcal{F}} & \leq\|f\|_{2}\,\|(\mathcal{N}+1)^{1/2}\psi\|_{\mathcal{F}},\\
		\|\phi(f)\psi\|_{\mathcal{F}} & \leq2\|f\|_{2}\|\left(\mathcal{N}+1\right)^{1/2}\psi\|_{\mathcal{F}}\,.
		\end{split}
		\label{eq:bd-a}
		\end{equation}
		Moreover, for any bounded one-particle operator $J$ on $L^{2}(\mathbb{R}^{3},dx)$
		and for every $\psi\in D(\mathcal{N})$, we find 
		\begin{equation}
		\|d\Gamma(J)\psi\|_{\mathcal{F}}\leq\|J\|_{\mathrm{op}}\|\mathcal{N}\psi\|_{\mathcal{F}}\,.\label{eq:J-bd}
		\end{equation}
	\end{lem}
	
	To consider the problem embedded into the Fock space, we extend
	Hamiltonian in \eqref{eq:N_body_Hamiltonian} to the Fock space by
	\begin{equation}
	\mathcal{H}_{N}:=\int\mathrm{d}x\,{\nabla_x a_{x}^{*}\nabla_x a_{x}}+\frac{1}{2N}\int\mathrm{d}x\mathrm{d}y\,V\left(x-y\right)a_{x}^{*}a_{y}^{*}a_{y}a_{x}.\label{eq:Fock_space_Hamiltonian}
	\end{equation}
	This definition satisfies $(\mathcal{H}_{N}\psi)^{(N)}=H_{N}\psi^{(N)}$
	for $\psi\in\mathcal{F}$. Hence it is a generalization of \eqref{eq:N_body_Hamiltonian} into the Fock space. The one-particle marginal density $\gamma_{\psi}^{\left(1\right)}$
	associated with $\psi$ is 
	\begin{equation}
	\gamma_{\psi}^{\left(1\right)}\left(x;y\right)=\frac{1}{\left\langle \psi,\mathcal{N}\psi\right\rangle_{\mathcal{F}} }\left\langle \psi,a_{y}^{*}a_{x}\psi\right\rangle_{\mathcal{F}} .\label{eq:Kernel_gamma}
	\end{equation}
	Note that $\gamma_{\psi}^{\left(1\right)}$ is a trace class operator
	on $L^{2}\left(\mathbb{R}^{3}\right)$ and $\text{Tr }\gamma_{\psi}^{\left(1\right)}=1$.
	It can be easily checked that \eqref{eq:Kernel_gamma} is equivalent
	to \eqref{eq:Kernel_of_Marginal_Density}.
	
	We defined a coherent state which is an eigenvector of annihilation operator $a(f)$ such that
	\[
	\psi\left(f\right)=e^{-\left\Vert f\right\Vert^{2}_{2}/2}\sum_{n\geq0}\frac{\left(a^{*}\left(f\right)\right)^{n}}{n!}\Omega=e^{-\left\Vert f\right\Vert^{2}_{2}/2}\sum_{n\geq0}\frac{1}{\sqrt{n!}}f^{\otimes n}.
	\]
	For $f\in L^{2}\left(\mathbb{R}^{3}\right)$, the
	Weyl operator $W\left(f\right)$ is defined by 
	\[
	W\left(f\right):=\exp\left(a^{*}\left(f\right)-a\left(f\right)\right)
	\]
	and it also satisfies 
	\[
	W\left(f\right)=e^{-\left\Vert f\right\Vert^{2}_{2}/2}\exp\left(a^{*}\left(f\right)\right)\exp\left(-a\left(f\right)\right),
	\]
	which is known as the Hadamard lemma in Lie algebra. 
	The Weyl operator is closely related to the coherent states.
	The coherent state can also be expressed in terms of the Weyl operator as 
	\begin{equation}
	\psi\left(f\right)=W\left(f\right)\Omega=e^{-\left\Vert f\right\Vert^{2}_{2}/2}\exp\left(a^{*}\left(f\right)\right)\Omega=e^{-\left\Vert f\right\Vert^{2}_{2}/2}\sum_{n\geq0}\frac{1}{\sqrt{n!}}f^{\otimes n}.\label{Weyl_f}
	\end{equation}
	
	We collect the useful properties of the Weyl operator and the coherent
	states in the following lemma. 
	\begin{lem}[Part of Lemma 2.2 in \cite{Chen2011a}]
		\label{lem:Basic_Weyl} Let $f,g\in L^{2}(\mathbb{R}^{3},\mathrm{d}x)$. 
		\begin{enumerate}
			\item The commutation relation between the Weyl operators is given by 
			\[
			W\left(f\right)W\left(g\right)=W\left(g\right)W\left(f\right)e^{-2\mathrm{i}\cdot\mathrm{Im}\left\langle f,g\right\rangle }=W\left(f+g\right)e^{-\mathrm{i}\cdot\mathrm{Im}\left\langle f,g\right\rangle }.
			\]
			\item The Weyl operator is unitary and satisfies that 
			\[
			W\left(f\right)^{*}=W\left(f\right)^{-1}=W\left(-f\right).
			\]
			\item The coherent states are eigenvectors of annihilation operators, i.e.,
			\[
			a_{x}\psi\left(f\right)=f\left(x\right)\psi\left(f\right)\quad\Rightarrow\quad a\left(g\right)\psi\left(f\right)=\left\langle g,f\right\rangle _{L^{2}}\psi\left(f\right).
			\]
			The commutation relation between the Weyl operator and the annihilation
			operator (or the creation operator) is thus 
			\[
			W^{*}\left(f\right)a_{x}W\left(f\right)=a_{x}+f\left(x\right)\quad\text{and}\quad W^{*}\left(f\right)a_{x}^{*}W\left(f\right)=a_{x}^{*}+\overline{f\left(x\right)}.
			\]
			\item The distribution of $\mathcal{N}$ with respect to the coherent state
			$\psi\left(f\right)$ is Poisson. In particular, 
			\[
			\left\langle \psi\left(f\right),\mathcal{N}\psi\left(f\right)\right\rangle_{\mathcal{F}} =\|f\|^{2}_{2},\qquad\left\langle \psi\left(f\right),\mathcal{N}^{2}\psi\left(f\right)\right\rangle_{\mathcal{F}} -\left\langle \psi\left(f\right),\mathcal{N}\psi\left(f\right)\right\rangle^{2}_{\mathcal{F}}=\|f\|^{2}_{2}.
			\]
		\end{enumerate}
	\end{lem}
	
	We define, for following lemmas,
	\begin{equation}
	d_{N}:=\frac{\sqrt{N!}}{N^{N/2}e^{-N/2}}\label{eq:d_N}
	\end{equation}
	and note that $C^{-1}N^{1/4}\leq d_{N}\leq CN^{1/4}$ for some constant
	$C>0$ independent of $N$, which can be easily checked by using Stirling's
	formula.
	\begin{lem}
		\label{lem:coherent_all} There exists a constant $C>0$ independent
		of $N$ such that, for any $\varphi\in L^{2}(\mathbb{R}^{3})$ with
		$\|\varphi\|_{2}=1$, we have 
		\[
		\left\Vert (\mathcal{N}+1)^{-1/2}W^{*}(\sqrt{N}\varphi)\frac{(a^{*}(\varphi))^{N}}{\sqrt{N!}}\Omega\right\Vert_{\mathcal{F}} \leq\frac{C(t)}{d_{N}}.
		\]
	\end{lem}
	
	\begin{proof}
		See \cite[Lemma 6.3]{Chen2011}. 
	\end{proof}
	\begin{lem}
		\label{lem:coherent_even_odd} Let $P_{m}$ be the projection onto
		the $m$-particle sector of the Fock space $\mathcal{F}$ for a non-negative
		integer $m$. Then, for any non-negative integers $k\leq(1/2)N^{1/3}$,
		\[
		\left\Vert P_{2k}W^{*}(\sqrt{N}\varphi)\frac{(a^{*}(\varphi))^{N}}{\sqrt{N!}}\Omega\right\Vert_{\mathcal{F}} \leq\frac{2}{d_{N}}
		\]
		and 
		\[
		\left\Vert P_{2k+1}W^{*}(\sqrt{N}\varphi)\frac{(a^{*}(\varphi))^{N}}{\sqrt{N!}}\Omega\right\Vert_{\mathcal{F}} \leq\frac{2(k+1)^{3/2}}{d_{N}\sqrt{N}}.
		\]
	\end{lem}
	
	\begin{proof}
		See \cite[Lemma 7.2]{Lee2013}. 
	\end{proof}
	
	\section{Proof of Main Theorem\label{sec:Pf-of-Main-Thm}}
	
	In this section, we prove the main result of the paper, Theorem 1.1 following the same logic given in \cite{Chen2018}.
	
	\subsection{Unitary operators and their generators}
	
	We let
	\[
	\psi_t = e^{-\mathrm{i}\mathcal{H}_N t}\varphi^{\otimes N}
	\]
	so that $\psi_t$ is the time evolution of the factorized state
	$\varphi^{\otimes N}$ with respect to the Hamiltonian $\mathcal{H}_{N}$.
	Noting the definition of $k$-particle marginal density \eqref{eq:Kernel_of_Marginal_Density}, the one-particle marginal density associated with $\psi_t$ can be written as
	\begin{align}
	\gamma_{N,t}^{(1)} & =\frac{\left\langle e^{-\mathrm{i}\mathcal{H}_{N}t}\varphi^{\otimes N},a_{y}^{*}a_{x}e^{-\mathrm{i}\mathcal{H}_{N}t}\varphi^{\otimes N}\right\rangle_{\mathcal{F}} }{\left\langle e^{-\mathrm{i}\mathcal{H}_{N}t}\varphi^{\otimes N},\mathcal{N}e^{-\mathrm{i}\mathcal{H}_{N}t}\varphi^{\otimes N}\right\rangle_{\mathcal{F}} }=\frac{1}{N}\left\langle \varphi^{\otimes N},e^{\mathrm{i}\mathcal{H}_{N}t}a_{y}^{*}a_{x}e^{-\mathrm{i}\mathcal{H}_{N}t}\varphi^{\otimes N}\right\rangle_{\mathcal{F}} \nonumber \\
	& =\frac{1}{N}\left\langle \frac{\left(a^{*}(\varphi)\right)^{N}}{\sqrt{N!}}\Omega,e^{\mathrm{i}\mathcal{H}_{N}t}a_{y}^{*}a_{x}e^{-\mathrm{i}\mathcal{H}_{N}t}\frac{\left(a^{*}(\varphi)\right)^{N}}{\sqrt{N!}}\Omega\right\rangle_{\mathcal{F}} .\label{eq:marginal_factorized}
	\end{align}
	We want to argue that \eqref{eq:marginal_factorized} can be approximated by the one-particle marginal density associated with the coherent states.
	To use the coherent state, we expand $a_{y}^{*}a_{x}$ around $N\overline{\varphi_{t}(y)}\varphi_{t}(x)$.
	The expansion leads us to investigate
	\begin{align}
	& W^{*}(\sqrt{N}\varphi_{s})e^{\mathrm{i}\mathcal{H}_{N}\left(t-s\right)}(a_{x}-\sqrt{N}\varphi_{t}(x))e^{-\mathrm{i}\mathcal{H}_{N}\left(t-s\right)}W(\sqrt{N}\varphi_{s})\label{eq:introducing U}\\
	& =W^{*}(\sqrt{N}\varphi_{s})e^{\mathrm{i}\mathcal{H}_{N}\left(t-s\right)}W(\sqrt{N}\varphi_{t})a_{x}W^{*}(\sqrt{N}\varphi_{t})e^{-\mathrm{i}\mathcal{H}_{N}\left(t-s\right)}W(\sqrt{N}\varphi_{s}).\nonumber 
	\end{align}
	By differentiating $W^{*}(\sqrt{N}\varphi_{t})e^{-\mathrm{i}\mathcal{H}_{N}\left(t-s\right)}W(\sqrt{N}\varphi_{s})$ with respect to $t$ as in \cite{Chen2011a,Lee2013,Rodnianski2009}, we have
	\begin{align}
	&\mathrm{i}\partial_{t}W^{*}(\sqrt{N}\varphi_{t})e^{-\mathrm{i}\mathcal{H}_{N}\left(t-s\right)}W(\sqrt{N}\varphi_{s})\notag\\
	&\qquad=:\left(\sum_{k=0}^{4}\mathcal{L}_{k}(t)\right)W^{*}(\sqrt{N}\varphi_{t})e^{-\mathrm{i}\mathcal{H}_{N}\left(t-s\right)}W(\sqrt{N}\varphi_{s}),\label{eq:derivative decomposition}
	\end{align}
	where
	\begin{align}
	\mathcal{L}_{0}(t) & :=\frac{N}{2}\int_{s}^{t}\mathrm{d}\tau\int\mathrm{d}x(V*|\varphi_{\tau}|^{2})(x)|\varphi_{\tau}(x)|^{2},\nonumber \\
	\mathcal{L}_{1}(t) & =0,\nonumber \\
	\mathcal{L}_{2}(t) & :=\int\mathrm{d}x\,\nabla_x a_{x}^{*}\nabla_x a_{x}+\int\mathrm{d}x\,\left(V*\left|\varphi_{t}\right|^{2}\right)\left(x\right)a_{x}^{*}a_{x}\nonumber\\
	&\qquad+\int\mathrm{d}x\mathrm{d}y\,V\left(x-y\right)\overline{\varphi_{t}\left(x\right)}\varphi_{t}\left(y\right)a_{y}^{*}a_{x}\nonumber \\
	&\qquad+\frac{1}{2}\int\mathrm{d}x\mathrm{d}y\,V\left(x-y\right)\left(\varphi_{t}\left(x\right)\varphi_{t}\left(y\right)a_{x}^{*}a_{y}^{*}+\overline{\varphi_{t}\left(x\right)}\,\overline{\varphi_{t}\left(y\right)}a_{x}a_{y}\right),\label{eq:L_2}\\
	\mathcal{L}_{3}(t) & :=\frac{1}{\sqrt{N}}\int\mathrm{d}x\mathrm{d}y\,V\left(x-y\right)\left(\varphi_{t}\left(y\right)a_{x}^{*}a_{y}^{*}+\overline{\varphi_{t}\left(y\right)}a_{x}^{*}a_{y}\right)a_{x},\text{ and}\label{eq:L_3}\\
	\mathcal{L}_{4} & :=\frac{1}{2N}\int\mathrm{d}x\mathrm{d}y\,V\left(x-y\right)a_{x}^{*}a_{y}^{*}a_{x}a_{y}.\label{eq:L_4}
	\end{align}
	Because the phase factor $\mathcal{L}_{0}(t)$ is just a complex-valued
	function, we can cancel this term by multiplying the right-hand
	side of \eqref{eq:derivative decomposition} by a function $e^{-\mathrm{i}\mathcal{L}_{0}(t)}$
	(see Section 3 of \cite{Lee2013}).
	Thus, if we define the unitary
	operator $\mathcal{U}(t;s)$ by 
	\[
	\mathcal{U}(t;s):=e^{-\mathrm{i}\omega(t;s)}W^{*}(\sqrt{N}\varphi_{t})e^{-\mathrm{i}\mathcal{H}_{N}\left(t-s\right)}W(\sqrt{N}\varphi_{s})
	\]
	with the phase factor 
	\[
	\omega(t;s):=\frac{N}{2}\int_{s}^{t}\mathrm{d}\tau\int\mathrm{d}x(V*|\varphi_{\tau}|^{2})(x)|\varphi_{\tau}(x)|^{2},
	\]
	then 
	\begin{equation}
	\mathrm{i}\partial_{t}\mathcal{U}\left(t;s\right)=\left(\mathcal{L}_{2}+\mathcal{L}_{3}+\mathcal{L}_{4}\right)\mathcal{U}\left(t;s\right)\quad\text{and}\quad\mathcal{U}\left(s;s\right)=I\label{eq:def_mathcalU}
	\end{equation}
	and 
	\[
	W^{*}(\sqrt{N}\varphi_{s})e^{\mathrm{i}\mathcal{H}_{N}\left(t-s\right)}\left(a_{x}-\sqrt{N}\varphi_{t}\left(x\right)\right)e^{-\mathrm{i}\mathcal{H}_{N}\left(t-s\right)}W(\sqrt{N}\varphi_{s})=\mathcal{U}^{*}\left(t;s\right)a_{x}\,\mathcal{U}\left(t;s\right).
	\]
	Let $\widetilde{\mathcal{L}}=\mathcal{L}_{2}+\mathcal{L}_{4}$ and
	define the unitary operator $\widetilde{\mathcal{U}}\left(t;s\right)$
	by 
	\begin{equation}
	\mathrm{i}\partial_{t}\widetilde{\mathcal{U}}\left(t;s\right)=\widetilde{\mathcal{L}}\left(t\right)\widetilde{\mathcal{U}}\left(t;s\right)\quad\text{and }\quad\widetilde{\mathcal{U}}\left(s;s\right)=1.\label{eq:def_mathcaltildeU}
	\end{equation}
	Since $\widetilde{\mathcal{L}}$ does not change the parity of the
	number of particles, 
	\begin{equation}
	\left\langle \Omega,\widetilde{\mathcal{U}}^{*}\left(t;0\right)a_{y}\,\widetilde{\mathcal{U}}\left(t;0\right)\Omega\right\rangle_{\mathcal{F}} =\left\langle \Omega,\widetilde{\mathcal{U}}^{*}\left(t;0\right)a_{x}^{*}\,\widetilde{\mathcal{U}}\left(t;0\right)\Omega\right\rangle_{\mathcal{F}} =0\label{eq:Parity_Consevation}
	\end{equation}
	We refer to Lemma 8.2 in \cite{Lee2013} for a rigorous proof of \eqref{eq:Parity_Consevation}.
	
	\subsection{Proof of Theorem \ref{thm:classical_Main}}
	
	As explained in Section \ref{sec:intro}, we use the technique developed
	in \cite{Lee2013} to prove Theorem \ref{thm:classical_Main}. The
	proof of Theorem \ref{thm:classical_Main} consists of the following
	two propositions. 
	
	\begin{prop}
		\label{prop:Et1} Suppose that the assumptions in Theorem \ref{thm:classical_Main}
		hold. For a Hermitian operator $J$ on $L^{2}(\mathbb{R}^{3})$, let
		\[
		E_{t}^{1}(J):=\frac{d_{N}}{N}\left\langle W^{*}(\sqrt{N}\varphi)\frac{(a^{*}(\varphi))^{N}}{\sqrt{N!}}\Omega,\mathcal{U}^{*}(t)d\Gamma(J)\mathcal{U}(t)\Omega\right\rangle_{\mathcal{F}} 
		\]
		Then, there exist a constant $C(t)$ depending only on $\lambda$, $\varphi_{0}$, and $t$ such that 
		\[
		\left|E_{t}^{1}(J)\right|\leq\frac{C(t)\|J\|_{\mathrm{op}}}{N}.
		\]
	\end{prop}

	\begin{prop}
		\label{prop:Et2} Suppose that the assumptions in Theorem \ref{thm:classical_Main}
		hold. For a Hermitian operator $J$ on $L^{2}(\mathbb{R}^{3})$, let
		\[
		E_{t}^{2}(J):=\frac{d_{N}}{\sqrt{N}}\left\langle W^{*}(\sqrt{N}\varphi)\frac{(a^{*}(\varphi))^{N}}{\sqrt{N!}}\Omega,\mathcal{U}^{*}(t)\phi(J\varphi_{t})\mathcal{U}(t)\Omega\right\rangle_{\mathcal{F}}
		\]
		Then, there exist a constant $C(t)$ depending only on $\lambda$, $\varphi_{0}$, and $t$ such that 
		\[
		\left|E_{t}^{2}(J)\right|\leq\frac{C(t)\|J\|_{\mathrm{op}}}{N}.
		\]
	\end{prop}
	
	Proof of Propositions \ref{prop:Et1} and \ref{prop:Et2} will be
	given later in section \ref{sec:Pf-of-Props}. With Propositions \ref{prop:Et1}
	and \ref{prop:Et2}, we now prove Theorem \ref{thm:classical_Main}.
	\begin{proof}[Proof of Theorem \ref{thm:classical_Main}]
		By the definition of $k$-particle density, in \eqref{eq:marginal_factorized} we have
		\[
		\gamma_{N,t}^{(1)}=\frac{1}{N}\left\langle \frac{\left(a^{*}(\varphi)\right)^{N}}{\sqrt{N!}}\Omega,e^{i\mathcal{H}_{N}t}a_{y}^{*}a_{x}e^{-i\mathcal{H}_{N}t}\frac{\left(a^{*}(\varphi)\right)^{N}}{\sqrt{N!}}\Omega\right\rangle_{\mathcal{F}}.
		\]
		From \eqref{eq:creation}, the factorized state $\varphi^{\otimes N}$ in $\mathcal{F}$ can be written in the following form:
		\begin{equation}
		\{0,0,\dots,0,\varphi^{\otimes N},0,\dots\}=\frac{\left(a^{*}(\varphi)\right)^{N}}{\sqrt{N!}}\Omega.\label{eq:coherent_vec}
		\end{equation}
		From \eqref{eq:projection} and \eqref{Weyl_f},
		we find that 
		\[
		\frac{\left(a^{*}(\varphi)\right)^{N}}{\sqrt{N!}}\Omega=\frac{\sqrt{N!}}{N^{N/2}e^{-N/2}}P_{N}W(\sqrt{N}\varphi)\Omega=d_{N}P_{N}W(\sqrt{N}\varphi)\Omega.
		\]
		Since $\left[\mathcal{H}_{N},\mathcal{N}\right]=0$,
		we also have that 
		\begin{align*}
		\gamma_{N,t}^{(1)}(x;y) & =\frac{1}{N}\left\langle \frac{\left(a^{*}(\varphi)\right)^{N}}{\sqrt{N!}}\Omega,e^{\mathrm{i}\mathcal{H}_{N}t}a_{y}^{*}a_{x}e^{-\mathrm{i}\mathcal{H}_{N}t}\frac{\left(a^{*}(\varphi)\right)^{N}}{\sqrt{N!}}\Omega\right\rangle_{\mathcal{F}} \\
		& =\frac{d_{N}}{N}\left\langle \frac{\left(a^{*}(\varphi)\right)^{N}}{\sqrt{N!}}\Omega,e^{\mathrm{i}\mathcal{H}_{N}t}a_{y}^{*}a_{x}e^{-\mathrm{i}\mathcal{H}_{N}t}P_{N}W(\sqrt{N}\varphi)\Omega\right\rangle_{\mathcal{F}} \\
		& =\frac{d_{N}}{N}\left\langle \frac{\left(a^{*}(\varphi)\right)^{N}}{\sqrt{N!}}\Omega,P_{N}e^{\mathrm{i}\mathcal{H}_{N}t}a_{y}^{*}a_{x}e^{-\mathrm{i}\mathcal{H}_{N}t}W(\sqrt{N}\varphi)\Omega\right\rangle_{\mathcal{F}} \\
		& =\frac{d_{N}}{N}\left\langle \frac{\left(a^{*}(\varphi)\right)^{N}}{\sqrt{N!}}\Omega,e^{\mathrm{i}\mathcal{H}_{N}t}a_{y}^{*}a_{x}e^{-\mathrm{i}\mathcal{H}_{N}t}W(\sqrt{N}\varphi)\Omega\right\rangle_{\mathcal{F}} .\\
		\end{align*}
		Moreover, using
		\[
		e^{\mathrm{i}\mathcal{H}_{N}t}a_{x}e^{-\mathrm{i}\mathcal{H}_{N}t}=W(\sqrt{N}\varphi)\mathcal{U}^{*}(t)(a_{x}+\sqrt{N}\varphi_{t}(x))\mathcal{U}(t)W^{*}(\sqrt{N}\varphi)
		\]
		and similar relation for the $a_x^*$, we obtain that
		\begin{align*}
		\gamma_{N,t}^{(1)}(x;y) & =\frac{d_{N}}{N}\left\langle \frac{\left(a^{*}(\varphi)\right)^{N}}{\sqrt{N!}}\Omega,e^{\mathrm{i}\mathcal{H}_{N}t}a_{y}^{*}a_{x}e^{-\mathrm{i}\mathcal{H}_{N}t}W(\sqrt{N}\varphi)\Omega\right\rangle_{\mathcal{F}} \\
		& =\frac{d_{N}}{N}\left\langle \frac{\left(a^{*}(\varphi)\right)^{N}}{\sqrt{N!}}\Omega,W(\sqrt{N}\varphi)\mathcal{U}^{*}(t)(a_{y}^{*}+\sqrt{N}\,\overline{\varphi_{t}\left(y\right)})(a_{x}+\sqrt{N}\varphi_{t}(x))\mathcal{U}(t)\Omega\right\rangle_{\mathcal{F}} .
		\end{align*}
		Hence, 
		\begin{align*}
		\gamma_{N,t}^{(1)}(x;y)-\overline{\varphi_{t}\left(y\right)}{\varphi_t(x)} & =\frac{d_{N}}{N}\left\langle \frac{\left(a^{*}(\varphi)\right)^{N}}{\sqrt{N!}}\Omega,W(\sqrt{N}\varphi)\mathcal{U}^{*}(t)a_{y}^{*}a_{x}\mathcal{U}(t)\Omega\right\rangle_{\mathcal{F}} \\
		& \quad+\overline{\varphi_{t}\left(y\right)}\frac{d_{N}}{\sqrt{N}}\left\langle \frac{\left(a^{*}(\varphi)\right)^{N}}{\sqrt{N!}}\Omega,W(\sqrt{N}\varphi)\mathcal{U}^{*}(t)a_{x}\mathcal{U}(t)\Omega\right\rangle_{\mathcal{F}} \\
		& \quad+\varphi_{t}(x)\frac{d_{N}}{\sqrt{N}}\left\langle \frac{\left(a^{*}(\varphi)\right)^{N}}{\sqrt{N!}}\Omega,W(\sqrt{N}\varphi)\mathcal{U}^{*}(t)a_{y}^{*}\mathcal{U}(t)\Omega\right\rangle_{\mathcal{F}} .
		\end{align*}
		By the definition of $E_{t}^{1}(J)$ and $E_{t}^{2}(J)$ in Propositions
		\ref{prop:Et1} and \ref{prop:Et2}, for any compact one-particle
		Hermitian operator $J$ on $L^{2}(\mathbb{R}^{3})$, we obtain
		\begin{align*}
		\operatorname{Tr}(J(\gamma_{N,t}^{(1)}-\left|\varphi_{t}\right\rangle \left\langle \varphi_{t}\right|) & =\int\mathrm{d}x\mathrm{d}yJ(x;y)\left(\gamma_{N,t}^{(1)}(y;x)-\varphi_{t}(y)\overline{\varphi_{t}\left(x\right)}\right)\\
		& =\frac{d_{N}}{N}\left\langle \frac{\left(a^{*}(\varphi)\right)^{N}}{\sqrt{N!}}\Omega,W(\sqrt{N}\varphi)\mathcal{U}^{*}(t)d\Gamma(J)\mathcal{U}(t)\Omega\right\rangle_{\mathcal{F}} \\
		& \quad+\frac{d_{N}}{\sqrt{N}}\left\langle \frac{\left(a^{*}(\varphi)\right)^{N}}{\sqrt{N!}}\Omega,W(\sqrt{N}\varphi)\mathcal{U}^{*}(t)\phi(J\varphi_{t})\mathcal{U}(t)\Omega\right\rangle_{\mathcal{F}} \\
		& =E_{t}^{1}(J)+E_{t}^{2}(J).
		\end{align*}
		Thus, Propositions \ref{prop:Et1} and \ref{prop:Et2} lead us
		that 
		\[
		\left|\operatorname{Tr}J(\gamma_{N,t}^{(1)}-\left|\varphi_{t}\right\rangle \left\langle \varphi_{t}\right|)\right|\leq C(t)\frac{\left\Vert J\right\Vert_{\mathrm{op}} }{N}.
		\]
		Since the space of compact operators is the dual to that of the trace
		class operators, and since $\gamma_{N,t}^{(1)}$ and $\left|\varphi_{t}\right\rangle \left\langle \varphi_{t}\right|$
		are Hermitian, 
		\[
		\operatorname{Tr}\left|\gamma_{N,t}^{(1)}-\left|\varphi_{t}\right\rangle \left\langle \varphi_{t}\right|\right|\leq\frac{C(t)}{N}
		\]
		which concludes the proof of Theorem \ref{thm:classical_Main}. 
	\end{proof}
	
	\section{Comparison of Dynamics and Proof of Propositions}

	\subsection{Comparison of dynamics}\label{sec:comparison}
	
	This section follows \cite{Rodnianski2009}. Rodnianski and Schlein
	used Hardy inequality $\sup_{x}\|V(\cdot-x)\varphi_{t}\|_{2}\leq C$
	in \cite{Rodnianski2009}. In \cite{Chen2018}, the authors used Stricharz
	estimate to bound the time integration of $\sup_{x}\|V(\cdot-x)\varphi_{t}\|_{2}$,
	i.e., \[
	\int_{0}^{t}\mathrm{d}s\sup_{x}\|V(\cdot-x)\varphi_{s}\|_{2}\leq\left(\int_{0}^{t}\mathrm{d}s\right)^{1/2}\sup_{x}\|V(\cdot-x)\varphi_{s}\|_{L^{2}((0,t),L^{\infty}(\mathbb{R}^{3}))}\leq Ct^{3/2}.
	\]
	This section will bound $\sup_{x}\|V(\cdot-x)\varphi_{t}\|_{2}$
	by $C(t)$ so that we can use the Table \ref{table:RoC_t}. Since
	the structure of each proof coincides with previous results \cite{Chen2018,Rodnianski2009},
	here we just provide the lemmas without proofs, because one can easily
	change all the $Ce^{Kt}$ appeared in \cite{Chen2018} by $C(t)$.
	\begin{lem}
		\label{lem:N_1_L3-1} Suppose that the assumptions in Theorem \ref{thm:classical_Main}
		hold. Then, for any $\psi\in\mathcal{F}$ and $j\in\mathbb{N}$, there
		exist a constant $C\equiv C(j)$ such that
		\[
		\left\Vert \left(\mathcal{N}+1\right)^{j/2}\mathcal{L}_{3}(t)\psi\right\Vert_{\mathcal{F}} \leq\frac{C}{\sqrt{N}}\sup_{x}\|V(x-\cdot)\varphi_{t}\|_{2}\left\Vert \left(\mathcal{N}+1\right)^{\left(j+3\right)/2}\psi\right\Vert_{\mathcal{F}} .
		\]
	\end{lem}
	\begin{proof}
		See Lemma 4.6 of~\cite{Chen2018}.
	\end{proof}
	
	\begin{lem}
		\label{lem:NjU} Suppose that the assumptions in Theorem \ref{thm:classical_Main}
		hold. Let $\mathcal{U}\left(t;s\right)$ be the unitary evolution
		defined in \eqref{eq:def_mathcalU}. Then for any $\psi\in\mathcal{F}$
		and $j\in\mathbb{N}$, there exist constants $C(t)\equiv C(t,j)$
		such that 
		\[
		\left\langle \mathcal{U}\left(t;s\right)\psi,\mathcal{N}^{j}\mathcal{U}\left(t;s\right)\psi\right\rangle_{\mathcal{F}} \leq C(t)\left\langle \psi,\left(\mathcal{N}+1\right)^{2j+2}\psi\right\rangle_{\mathcal{F}} .
		\]
	\end{lem}
	\begin{proof}
		See Lemma 4.1 of~\cite{Chen2018}.
	\end{proof}
	
	\begin{lem}
		\label{lem:tildeNj} Suppose that the assumptions in Theorem \ref{thm:classical_Main}
		hold. 
		Let $\widetilde{\mathcal{U}}\left(t;s\right)$ be the unitary evolution
		defined in \eqref{eq:def_mathcaltildeU}.
		Then, for any $\psi\in\mathcal{F}$ and $j\in\mathbb{N}$, there
		exist a constant $C(t)\equiv C(t,j)$ such that
		\[
		\left\langle \widetilde{\mathcal{U}}\left(t;s\right)\psi,\mathcal{N}^{j}\widetilde{\mathcal{U}}\left(t;s\right)\psi\right\rangle_{\mathcal{F}} \leq C(t)\left\langle \psi,\left(\mathcal{N}+1\right)^{2j+2}\psi\right\rangle_{\mathcal{F}} .
		\]
	\end{lem}
	\begin{proof}
		See Lemma 4.5 of~\cite{Chen2018}.
	\end{proof}
	
	The following lemma will be used in the proof of Proposition \ref{prop:Et2} in the following Section \ref{sec:Pf-of-Props}.
	\begin{lem}
		\label{lem:NjUphiUtildeUphitildeU} Suppose that the assumptions in
		Theorem \ref{thm:classical_Main} hold.
		Let $\mathcal{U}\left(t;s\right)$ and $\widetilde{\mathcal{U}}\left(t;s\right)$ be the unitary evolution
		defined in \eqref{eq:def_mathcalU} and \eqref{eq:def_mathcaltildeU} respectively.
		Then, for all $j\in\mathbb{N}$,
		there exist constants $C(t)\equiv C(t,j)$ such that, for any $f\in L^{2}(\mathbb{R}^{3})$,
		\[
		\left\Vert \left(\mathcal{N}+1\right)^{j/2}\left(\mathcal{U}^{*}\left(t\right)\phi(f)\mathcal{U}\left(t\right)-\mathcal{\widetilde{U}}^{*}\left(t\right)\phi(f)\mathcal{\widetilde{U}}\left(t\right)\right)\Omega\right\Vert_{\mathcal{F}} \leq C(t)\frac{\|f\|_{2}}{\sqrt{N}}.
		\]
	\end{lem}
	\begin{proof}
		See Lemma 4.7 of~\cite{Chen2018}. To prove this lemma, we use Lemma \ref{lem:N_1_L3-1} and \ref{lem:tildeNj} as Lemma 4.7 of \cite{Chen2018} used Lemma 4.1 and 4.5 of \cite{Chen2018}.
	\end{proof}
	
	\subsection{Proof of Propositions \ref{prop:Et1} and \ref{prop:Et2}}\label{sec:Pf-of-Props-1}\label{sec:Pf-of-Props}
	
	In this section, we prove Propositions \ref{prop:Et1} and \ref{prop:Et2}
	by applying the lemmas provided in Subsection \ref{sec:comparison}.
	
	\begin{proof}[Proof of Proposition \ref{prop:Et1}]
		Note that
		\[
		E_{t}^{1}(J)=\frac{d_{N}}{N}\left\langle W^{*}(\sqrt{N}\varphi)\frac{(a^{*}(\varphi))^{N}}{\sqrt{N!}}\Omega,\mathcal{U}^{*}(t)d\Gamma(J)\mathcal{U}(t)\Omega\right\rangle_{\mathcal{F}}
		\]
		implies that
		\begin{align}
		\left|E_{t}^{1}(J)\right| & =\left|\frac{d_{N}}{N}\left\langle W^{*}(\sqrt{N}\varphi)\frac{(a^{*}(\varphi))^{N}}{\sqrt{N!}}\Omega,\mathcal{U}^{*}(t)d\Gamma(J)\mathcal{U}(t)\Omega\right\rangle_{\mathcal{F}} \right|\label{eq:E_t^1 1-1}\\
		& \leq\frac{d_{N}}{N}\left\Vert (\mathcal{N}+1)^{-\frac{1}{2}}W^{*}(\sqrt{N}\varphi)\frac{(a^{*}(\varphi))^{N}}{\sqrt{N!}}\Omega\right\Vert_{\mathcal{F}} \nonumber\\
		&\qquad\quad\times\left\Vert (\mathcal{N}+1)^{\frac{1}{2}}\mathcal{U}^{*}(t)d\Gamma(J)\mathcal{U}(t)\Omega\right\Vert_{\mathcal{F}}. \nonumber 
		\end{align}
		Using Lemma \ref{lem:coherent_all}, we have
		\begin{equation}
		\left\Vert (\mathcal{N}+1)^{-\frac{1}{2}}W^{*}(\sqrt{N}\varphi)\frac{(a^{*}(\varphi))^{N}}{\sqrt{N!}}\Omega\right\Vert_{\mathcal{F}} \leq\frac{C(t)}{d_{N}}.\label{eq:E_t^1 2-1}
		\end{equation}
		By applying Lemma \ref{lem:NjU} and \eqref{eq:J-bd} several times, we get 
		\begin{align}
		\left\Vert (\mathcal{N}+1)^{\frac{1}{2}}\mathcal{U}^{*}(t)d\Gamma(J)\mathcal{U}(t)\Omega\right\Vert_{\mathcal{F}}  & \leq C(t)\left\Vert (\mathcal{N}+1)^{2}d\Gamma(J)\mathcal{U}(t)\Omega\right\Vert_{\mathcal{F}}\notag\\
		&\leq C(t)\left\Vert J\right\Vert_{\mathrm{op}} \left\Vert (\mathcal{N}+1)^{3}\mathcal{U}(t)\Omega\right\Vert_{\mathcal{F}} \nonumber \\
		& \leq C(t)\left\Vert J\right\Vert_{\mathrm{op}} \left\Vert (\mathcal{N}+1)^{7}\Omega\right\Vert_{\mathcal{F}} .\label{eq:E_t^1 3-1}
		\end{align}
		Therefore, from \eqref{eq:E_t^1 1-1}, \eqref{eq:E_t^1 2-1}, and \eqref{eq:E_t^1 3-1}, we have the desired bound
		\[
		\left|E_{t}^{1}(J)\right|\leq\frac{C(t)\|J\|_{\mathrm{op}}}{N}.
		\]
	\end{proof}
	For the proof of Proposition \ref{prop:Et2}, We apply a very similar approach to the one used in the proof of Lemma 4.2 in \cite{Lee2013}. To obtain the logical completeness, we fill the detail.
	\begin{proof}[Proof of Proposition \ref{prop:Et2}]
		
		Let
		\[
		\mathcal{R}(f)=\mathcal{U}^{*}(t)\phi(f)\mathcal{U}(t)-\widetilde{\mathcal{U}}^{*}(t)\phi(f)\widetilde{\mathcal{U}}(t).
		\]
		According to \eqref{eq:Parity_Consevation}, the even sector will have zero amplitude, i.e. 
		\[
		P_{2\ell}\,\widetilde{\mathcal{U}}^{*}(t)\phi(J\varphi_{t})\widetilde{\mathcal{U}}(t)\Omega=0
		\]
		for all $\ell=0,1,{\dots}$. (See Lemma 8.2 in \cite{Lee2013} for more
		detail.) This gives us that
		\begin{align}
		\left|E_{t}^{2}(J)\right| & =\frac{d_{N}}{\sqrt{N}}\left\langle \frac{(a^{*}(\varphi))^{N}}{\sqrt{N!}}\Omega,W^{*}(\sqrt{N}\varphi)\mathcal{\widetilde{\mathcal{U}}}^{*}(t)\phi(J\varphi_{t})\mathcal{\widetilde{\mathcal{U}}}(t)\Omega\right\rangle_{\mathcal{F}} \nonumber \\
		& \qquad+\frac{d_{N}}{\sqrt{N}}\left\langle \frac{(a^{*}(\varphi))^{N}}{\sqrt{N!}}\Omega,W^{*}(\sqrt{N}\varphi)\mathcal{R}(J\varphi_{t})\Omega\right\rangle_{\mathcal{F}} \nonumber \\
		& \leq\frac{d_{N}}{\sqrt{N}}\left\Vert \sum_{\ell=1}^{\infty}(\mathcal{N}+1)^{-\frac{5}{2}}P_{2\ell-1}W^{*}(\sqrt{N}\varphi)\frac{(a^{*}(\varphi))^{N}}{\sqrt{N!}}\Omega\right\Vert_{\mathcal{F}} \nonumber\\
		&\qquad\qquad\times \left\Vert (\mathcal{N}+1)^{\frac{5}{2}}\mathcal{\widetilde{\mathcal{U}}}^{*}(t)\phi(J\varphi_{t})\widetilde{\mathcal{U}}(t)\Omega\right\Vert_{\mathcal{F}} \nonumber \\
		& \qquad+\frac{d_{N}}{\sqrt{N}}\left\Vert (\mathcal{N}+1)^{-\frac{1}{2}}W^{*}(\sqrt{N}\varphi)\frac{(a^{*}(\varphi))^{N}}{\sqrt{N!}}\Omega\right\Vert_{\mathcal{F}} \nonumber\\
		&\qquad\qquad \times\left\Vert (\mathcal{N}+1)^{\frac{1}{2}}\mathcal{R}(J\varphi_{t})\Omega\right\Vert_{\mathcal{F}} \label{eq:e_t^2 1-1}
		\end{align}
		We divide the sum into two group using $L=\frac{1}{2}N^{1/3}$, Lemma \ref{lem:coherent_all},
		and Lemma \ref{lem:coherent_even_odd} such that
		\begin{align}
		& \left\Vert \sum_{\ell=1}^{\infty}(\mathcal{N}+1)^{-\frac{5}{2}}P_{2\ell-1}W^{*}(\sqrt{N}\varphi)\frac{(a^{*}(\varphi))^{N}}{\sqrt{N!}}\Omega\right\Vert^{2}_{\mathcal{F}}\nonumber\\
		& \qquad\leq\sum_{\ell=1}^{L}\left\Vert (\mathcal{N}+1)^{-\frac{5}{2}}P_{2\ell-1}W^{*}(\sqrt{N}\varphi)\frac{(a^{*}(\varphi))^{N}}{\sqrt{N!}}\Omega\right\Vert_{\mathcal{F}} ^{2}\nonumber\\
		& \qquad\qquad+\frac{1}{L^{4}}\sum_{\ell=L}^{\infty}\left\Vert (\mathcal{N}+1)^{-1/2}P_{2\ell-1}W^{*}(\sqrt{N}\varphi)\frac{(a^{*}(\varphi))^{N}}{\sqrt{N!}}\Omega\right\Vert^{2}_{\mathcal{F}}\nonumber\\
		& \qquad\leq\left(\sum_{\ell=1}^{L}\frac{C}{\ell^{2}d_{N}^{2}N}\right)+\frac{C}{N^{4/3}}\left\Vert (\mathcal{N}+1)^{-1/2}W^{*}(\sqrt{N}\varphi)\frac{(a^{*}(\varphi))^{N}}{\sqrt{N!}}\Omega\right\Vert_{\mathcal{F}} \leq\frac{C(t)}{d_{N}^{2}N}.\label{eq:E2bdd}
		\end{align}
		Applying Lemma \ref{lem:tildeNj}, 
		\begin{alignat*}{1}
		& \left\Vert (\mathcal{N}+1)^{\frac{5}{2}}\widetilde{\mathcal{U}}^{*}(t)\phi(J\varphi_{t})\widetilde{\mathcal{U}}(t)\Omega\right\Vert_{\mathcal{F}} \leq C(t)\left\Vert (\mathcal{N}+1)^{\frac{5}{2}}\phi(J\varphi_{t})\widetilde{\mathcal{U}}(t)\Omega\right\Vert_{\mathcal{F}} \\
		& \quad\leq C(t)\|J\varphi_{t}\|\left\Vert (\mathcal{N}+1)^{3}\mathcal{\widetilde{\mathcal{U}}}(t)\Omega\right\Vert_{\mathcal{F}} \leq C(t)\|J\|\left\Vert (\mathcal{N}+1)^{3}\Omega\right\Vert_{\mathcal{F}} \leq C\|J\|_{\mathrm{op}}.
		\end{alignat*}
		For the second term of \eqref{eq:E2bdd}, we apply Lemmas \ref{lem:coherent_all}
		and \ref{lem:NjUphiUtildeUphitildeU}, and put $J\varphi_{t}$ into  $f$.
		Altogether, we get the desired bound
		\[
		\left\Vert (\mathcal{N}+1)^{j/2}\mathcal{R}(f)\Omega\right\Vert_{\mathcal{F}} \leq\frac{C(t)\|f\|_{2}}{N}.
		\] 
	\end{proof}
	
	\section*{Acknowledgments}
	The author is grateful to numerous helpful discussions
	and suggestions from Ji Oon Lee.
	The author also would like to thank the anonymous referee for carefully reading the manuscript and providing helpful comments.
	This research is supported in part by KIA Motors Scholarship.

\end{document}